\documentclass[thesis,11pt,oneside]{amsart}%
\usepackage{amssymb}
\usepackage{amsmath}
\usepackage{mathrsfs}
\usepackage{amsfonts}
\usepackage{graphicx}
\usepackage[left=1.5in, right=1.5in, top=1in, bottom=1in, includefoot, headheight=13.6pt]{geometry}
\usepackage[all]{xy}
\newtheorem{theorem}{Theorem}

\theoremstyle{plain}

\newtheorem{corollary}{Corollary}

\newtheorem*{lemma}{Lemma}

\newtheorem{proposition}{Proposition}

\newtheorem*{xrem}{Remark}

\numberwithin{equation}{section}
\begin{document}
\date{October 2010}
\keywords{Quantum Stochastic Calculus, Non-adapted Stochastic Integrals, White Noise Analysis, Quantum Stochastic Processes.}


\title[Quantum Stochastic Integrals and Differentials]{$\mathrm{Q}$-Adapted Quantum Stochastic Integrals and Differentials in Fock Scale}

\author{V. P. Belavkin}
\address{School of Mathematical Sciences, University Park\\
Nottingham, NG7 2RD, UK\\
E-mail: viacheslav.belavkin@nottingham.ac.uk}

\author{M. F. Brown}
\address{School of Mathematical Sciences, University Park\\
Nottingham, NG7 2RD, UK\\
E-mail: pmxmb1@nottingham.ac.uk}

\maketitle

\begin{abstract}
In this paper we first introduce the Fock-Guichardet formalism for the quantum stochastic (QS) integration, then the four fundamental processes of the dynamics are introduced in the canonical basis as the operator-valued measures, on a space-time $\sigma$-field $\mathfrak{F}_\mathbb{X}$, of the QS integration. Then rigorous analysis of the QS integrals is carried out, and continuity of the QS derivative $\mathbf{D}$ is proved. Finally, $\mathrm{Q}$-adapted dynamics is discussed, including Bosonic ($\mathrm{Q}=\mathrm{I}$), Fermionic ($\mathrm{Q}=-\mathrm{I}$), and monotone ($\mathrm{Q}=\mathrm{O}$) quantum dynamics. These may be of particular interest to quantum field theory, quantum open systems, and quantum theory of stochastic processes.
\end{abstract}

\section{Introduction}
Non-commutative generalization of the It\^{o} stochastic calculus, developed
in {\cite{AccF88}}, {\cite{AccQ89}}, {\cite{EvaH88}}, {\cite{LinM88a}}, {\cite{Mey87}} and {\cite{ParS86}} gave an
adequate mathematical tool for studying the behavior of open quantum dynamical
systems singularly interacting with a boson quantum-stochastic field. Quantum
stochastic calculus also made it possible to solve an old problem of
describing such systems with continuous observation and constructing a quantum
filtration theory which would explain a continuous spontaneous collapse under
the action of such observation {\cite{Be88a}}, {\cite{Be88c}} and {\cite{Be92}}. This gave
examples of stochastic non-unitary, non-stationary, and even non-adapted
evolution equations in a Hilbert space whose solution requires a proper
definition of chronologically ordered quantum stochastic semigroups, and
exponents of operators, by extending the notion of the multiple stochastic
integral to non-commuting objects.

Here is the first part of an outline of the solution to this important problem by developing a \emph{$Q$-adapted} form of the new quantum stochastic calculus constructed in {\cite{Be91}} in a natural scale of Fock spaces. It is based on an explicit definition, introduced in {\cite{Be90c}}, of the non-adapted quantum stochastic integral, as a non-commutative generalization of the Skorokhod integral {\cite{Sko75}} represented in the Fock space. The point derivative of the quantum stochastic calculus is discussed as an operator on the scaled Fock space, and consequentially the single integral operator-kernels are derived from the operator-kernels of the multiple stochastic integral. These quantum stochastic derivatives are then presented in an explicit Q-adapted form, and we recover the Fermionic anti-commutation relation as well as the trivial Bosonic commutator of the usual adapted process.

The approach used here is similar in spirit to the kernel calculus of
Maassen-Lindsay-Meyer {\cite{LinM88a}}, {\cite{Mey87}}, however the difference is that all
the main objects are constructed not in terms of kernels but in terms of
operators represented in the Fock space. In addition we employ a much more
general notion of multiple stochastic integral, non-adapted in general but focusing now on Q-adapted processes, which
reduces to the notion of the kernel representation of an operator only in the
case of a scalar (non-random) operator function under the integral. The
possibility of defining a non-adapted single integral in terms of the kernel
calculus was shown by Lindsay {\cite{Lin90}}, but the notion of the multiple
quantum-stochastic integral was introduced in {\cite{Be91}}.

\section{Rigged Guichardet-Fock Space}

Let $(\mathbb{X},\lambda)$ be an essentially ordered space, that is, a measurable space
$\mathbb{X}$ with a $\sigma$-finite measure $\lambda:\mathfrak{F}_{\mathbb{X}}\ni
\boldsymbol{\bigtriangleup}\mapsto\lambda\left(  \boldsymbol{\bigtriangleup
}\right)  \geq0$ and an ordering relation $x\leq x^{\prime}$ with the property
that any $n$-tuple $\left(  x_{1},\ldots,x_{n}\right)  \in \mathbb{X}^{n}$ can be
identified up to a permutation with a chain $\varkappa=\{x_{1}<\cdots<x_{n}\}$
modulo the product measure $\prod_{i=1}^{n}\mathrm{d}x_{i}$ of $\mathrm{d}%
x:=\lambda\left(  \mathrm{d}x\right)  $. In other words, we assume that the
measurable ordering is almost total, that is, for any $n$ the product measure
of $n$-tuples $\boldsymbol{s}\in \mathbb{X}^{n}$ with components $\left(  x_{1}%
,\ldots,x_{n}\right)  $ that are not comparable is zero. Hence, in particular,
it follows that the measure $\lambda$ on $\mathbb{X}$ is atomless and we may assume
that this essentially total ordering on $\mathbb{X}$ is induced from the linear order
in $\mathbb{R}_{+}$ by a measurable map $t:\mathbb{X}\rightarrow\mathbb{R}_{+}$
relatively to which $\lambda$ is absolutely continuous with respect to the
Lebesgue measure $\mathrm{d}t$ on $\mathbb{R}_{+}$ in the sense of admitting
the disintegration
\[
\left\langle f\circ t,1_{\boldsymbol{\bigtriangleup}}\right\rangle _{\lambda
}:=\int_{\boldsymbol{\bigtriangleup}}f(t(x))\lambda\left(  \mathrm{d}x\right)
=\int_{0}^{\infty}f(t)\lambda_{\boldsymbol{\bigtriangleup}}(t)\mathrm{d}%
t\equiv\left\langle f,\lambda_{\boldsymbol{\bigtriangleup}}\right\rangle .
\]
Here $1_{\boldsymbol{\bigtriangleup}}$ is the indicator of any integrable
subset $\boldsymbol{\bigtriangleup}\subseteq \mathbb{X}$ and $f$ is any essentially
bounded function $f:\mathbb{R}_{+}\rightarrow\mathbb{C}$ and
$\boldsymbol{\bigtriangleup}\mapsto\lambda_{\boldsymbol{\bigtriangleup}}(t)$
is defined by duality as a positive measure on $\mathbb{X}$ for each $t\in
\mathbb{R}_{+} $. In any case we will fix a map $t$ that the above condition
holds and $t(x)<t(x^{\prime})$ if $x<x^{\prime}$, interpreting $t(x)$ as the
time at the point $x\in \mathbb{X}$. For example, $t(x)=t$ for $x=(\vec{x},t)$ if
$\mathbb{X}=\mathbb{R}^{d}\times\mathbb{R}_{+}$ is the $(d+1)$-dimensional space-time
with the casual ordering {\cite{Be85}} and $\mathrm{d}x=\mathrm{d}\vec{x}%
\mathrm{d}t$, where $\mathrm{d}\vec{x}$ is the standard volume element on
$d$-dimensional space $\mathbb{R}^{d}\ni\vec{x}$.

We shall identify the finite chains $\varkappa$ with increasingly indexed $n
$-tuples $(x_{1},\ldots,x_{n})\equiv\boldsymbol{s}$ with $x_{i}\in \mathbb{X}$, $x_{1}%
<\cdots<x_{n}$, denoting by $\mathcal{X}=\sum_{n=0}^{\infty}\mathcal{X}_{n}$
the set of all finite chains as the union of the sets
\[
\mathcal{X}_{n}=\{\boldsymbol{s}\in \mathbb{X}^{n}:x_{1}<\cdots<x_{n}\}
\]
with one-element $\mathcal{X}_{0}=\{\emptyset\}$ containing the empty chain as
a subset of $\mathbb{X}$: $\emptyset=\mathbb{X}^{0}$. We introduce a measure `element'
$\mathrm{d}\varkappa=\prod_{x\in\varkappa}\mathrm{d}x$ on $\mathcal{X}$
induced by the direct sum $\oplus_{n=0}^{\infty}\lambda^{\otimes n}\left(
\boldsymbol{\bigtriangleup}_{n}\right)  ,\boldsymbol{\bigtriangleup}_{n}%
\in\mathfrak{F}_{\mathbb{X}}^{\otimes n}$ of product measures $\mathrm{d}%
\boldsymbol{s}=\prod_{i=1}^{n}\mathrm{d}x_{i}$ on $\mathbb{X}^{n}$ with the unit mass
$\mathrm{d}\varkappa=1$ at the only atomic point $\varkappa=\emptyset$.

Let $\{\mathfrak{k}_{x}:x\in \mathbb{X}\}$ be a family of Hilbert spaces $\mathfrak{k}%
_{x}$, let $\mathfrak{p}_{0}$ be an additive semigroup of nonnegative
essentially measurable locally bounded functions $q:\mathbb{X}\rightarrow\mathbb{R}%
_{+}$ with zero included $0\in\mathfrak{p}_{0}$, and let $\mathfrak{p}_{1}%
=\{1+q_{0}:q_{0}\in\mathfrak{p}_{0}\}$. For example, in the case
$\mathbb{X}=\mathbb{R}^{d}\times\mathbb{R}_{+}$ by $\mathfrak{p}_{1}$ we mean the set
of polynomials $q(x)=1+\sum_{k=0}^{m}c_{k}|\vec{x}|^{k}$ with respect to the
modulus $|\vec{x}|=(\Sigma x_{i}^{2})^{1/2}$ of a vector $\vec{x}\in
\mathbb{R}^{d}$ with coefficients $c_{k}\geq0$. We denote by
$\mathit{K}_{\star}(q)$ the Hilbert space of essentially measurable
vector-functions $\mathrm{k}:x\mapsto\mathrm{k}(x)\in\mathfrak{k}_{x}$ which
are square integrable with the weight $q\in\mathfrak{p}_{1}$:
\[
\left\Vert \mathrm{k}\right\Vert (q)=\left(  \int\left\Vert \mathrm{k}%
(x)\right\Vert _{x}^{2}q(x)\mathrm{d}x\right)  ^{1/2}<\infty.
\]
With $q\geq1$, any space $\mathit{K}_{\star}(q)$ can be embedded into the
Hilbert space $\mathcal{K}_\ast=\mathit{K}_{\star}(1)$, and the intersection
$\cap_{q\in\mathfrak{p}_{1}}\mathit{K}_{\star}(q)\subseteq\mathcal{K}_\ast$
can be identified with the projective limit $\mathit{K}_{+}=\lim
_{q\rightarrow\infty}\mathit{K}_{\star}(q)$. This follows from the facts that
the function $\left\Vert \mathrm{k}\right\Vert (q)$ is increasing: $q\leq
p\Rightarrow\left\Vert \mathrm{k}\right\Vert (q)\leq\left\Vert \mathrm{k}%
\right\Vert (p)$, and so $\mathit{K}_{\star}(p)\subseteq\mathit{K}_{\star}(q)$,
and that the set $\mathfrak{p}_{1}$ is directed in the sense that for any $q=1+r$
and $p=1+s$, $r,s\in\mathfrak{p}_{0}$, there is a function in $\mathfrak{p}%
_{1}$ majorizing $q$ and $p$ (we can take for example $q+p-1=1+r+s\in
\mathfrak{p}_{1}$). In the case of polynomials $q\in\mathfrak{p}_{1}$ on
$\mathbb{X}=\mathbb{R}^{d}\times\mathbb{R}_{+} $ the decreasing family $\{\mathit{K}%
_{\star}(q)\}$, where $\mathfrak{k}_{x}=\mathbb{C}$, is identical with the
integer Sobolev scale of vector fields $\mathrm{k}:\mathbb{R}^{d}%
\rightarrow\mathit{L}^{2}(\mathbb{R}_{+})$ with values $\mathrm{k}%
(x)(t)=\mathrm{k}(x,t)$ in the Hilbert space $\mathit{L}^{2}(\mathbb{R}_{+})$
of square integrable functions on $\mathbb{R}_{+}$. If we replace
$\mathbb{R}^{d}$ by $\mathbb{Z}^{d}$ and if we restrict ourselves to the
positive part of the integer lattice $\mathbb{Z}^{d}$, then we obtain the
Schwartz space in the form of vector fields $\mathrm{k}\in\mathit{K}_{+}$.

The dual space $\mathit{K}_{\star}^{-}$ to$\ \mathit{K}_{+}$ is the space of
generalized vector-functions $\mathrm{f}\left(  x\right)  $ defining the
continuous functionals
\[
\left\langle \mathrm{f}\mid\mathrm{k}\right\rangle =\int\left\langle
\mathrm{f}(x)\mid\mathrm{k}(x)\right\rangle \,\mathrm{d}x,\quad\mathrm{k}%
\in\mathit{K}_{+}.
\]
It is the inductive limit $\mathit{K}_{-}=\lim_{q\rightarrow0}\mathit{K}%
_{\star}(q)$ in the opposite scale $\{\mathit{K}_{\star}(q):q\in\mathfrak{p}%
_{-}\}$, where $\mathfrak{p}_{-}$ is the set of functions $q:\mathbb{X}\rightarrow
(0,1]$ such that $1/q\in\mathfrak{p}_{1}$, which is the union $\cup
_{q\in\mathfrak{p}_{-}}\mathit{K}_{\star}(q)$ of the inductive family of
Hilbert spaces $\mathit{K}_{\star}(q)$, $q\in\mathfrak{p}_{-}$, with the norms
$\left\Vert \mathrm{k}\right\Vert (q)$, containing as the minimal the space
$\mathcal{K}_{\ast}$. Thus we obtain the Gel'fand chain
\[
\mathit{K}_{+}\subseteq\mathit{K}_{\star}(q_{+})\subseteq\mathcal{K}_{\ast
}\subseteq\mathit{K}_{\star}(q_{-})\subseteq\mathit{K}_{-}%
\]
in the extended scale $\{\mathit{K}_{\star}(q):q\in\mathfrak{p}\}$, where
$\mathfrak{p}=\mathfrak{p}_{-}\cup\mathfrak{p}_{1}$, with $q_{+}%
\in\mathfrak{p}_{1},\,q_{-}\in\mathfrak{p}_{-}$. The dual space $\mathit{K}%
_{+}^{\star}=\mathit{K}^{-}$ is the space of the continuous linear functionals
on $\mathit{K}_{+}$ containing the Hilbert space $\mathcal{K}$ called
the\emph{\ rigged space} with respect to the dense subspace $\mathit{K}%
^{+}=\mathit{K}_{-}^{\star}$ of $\mathcal{K}$ equipped with the projective
convergence in the scale $\left\Vert \mathrm{k}^{\ast}\right\Vert \left(
q\right)  =\left\Vert \mathrm{k}\right\Vert \left(  q\right)  $ for
$q\in\mathfrak{p}_{1}$.

We can similarly define a Fock-Gel'fand triple $(\mathit{F}_{+},\mathcal{F}%
_{\ast},\mathit{F}_{-})$  with%
\[
\mathit{F}_{+}=\cap_{q\in\mathfrak{p}_{1}}\mathit{F}_{\star}(q),\;\mathcal{F}%
_{\ast}=\mathit{F}_{\star}(1),\;\mathit{F}_{-}=\cup_{q\in\mathfrak{p}_{-}%
}\mathit{F}_{\star}(q),
\]
for the Hilbert scale $\{\mathit{F}_{\star}%
(q):q\in\mathfrak{p}\}$ of the symmetric Fock spaces $\mathit{F}_{\star}(q)=\oplus_{n=0}^\infty K_\star^{(n)}(q)$ over
$\mathit{K}_{\star}(q)$, where $K_\star^{(0)}(q)=\mathbb{C}$, $K_\star^{(1)}(q)=K_\star(q)$, and each $K_\star^{(n)}(q)$ for $n>1$ is given by the product-weight $q_n(x_1,\ldots,x_n)=\prod_{i=1}^nq(x_i)$ on $\mathbb{X}^n$. We shall consider the Guichardet {\cite{Gui72}} representation of the symmetric tensor-functions $\psi_n\in K^{(n)}_\star(q)$ regarding them as the restrictions $\psi|\mathcal{X}_n$ of the functions $\psi:\varkappa\mapsto\psi(\varkappa)\in K_\star^\otimes(\varkappa)$ with sections in the Hilbert products $\mathit{K}_{\star}^{\otimes}(\varkappa)=\otimes
_{x\in\varkappa}\mathfrak{k}_{x}$, square integrable with the product weight
$q(\varkappa)=\prod_{x\in\varkappa}q(x)$:
\[
\Vert\psi\Vert(q)=\left(  \int\Vert\psi(\varkappa)\Vert^{2}q(\varkappa
)\,\mathrm{d}\varkappa\right)  ^{1/2}<\infty.
\]
The integral here is over all chains $\varkappa\in\mathcal{X}$ and defines the
pairing on $\mathit{F}_{+}$ by
\[
\left\langle \psi\mid\psi\right\rangle =\int\left\langle \psi(\varkappa
)\mid\psi(\varkappa)\right\rangle \,\mathrm{d}\varkappa,\quad\psi\in
\mathit{F}_{+}.
\]
In more detail we can write this in the form
\[
\int\Vert\psi(\varkappa)\Vert^{2}q(\varkappa)\mathrm{d}\varkappa=\sum
_{n=0}^{\infty}\int\limits_{0\leq t_{1}<}\cdots\int\limits_{<t_{n}<\infty
}\Vert\psi(x_{1},\ldots,x_{n})\Vert^{2}\prod_{i=1}^{n}q(x_{i})\mathrm{d}x_{i},
\]
where the $n$-fold integrals for $\psi_n\in K_\star^{(n)}$ are taken over simplex domains $\mathcal{X}%
_{n}=\{\boldsymbol{s}\in \mathbb{X}^{n}:t(x_{1})<\cdots<t(x_{n})\}$.

One can easily establish
an isomorphism between the space $\mathit{F}_{\star}(q)$ and the symmetric (or
antisymmetric) Fock space over $\mathit{K}_{\star}(q)$ with a nonatomic measure $\mathrm{d}x$ in $\mathbb{X}$. It is defined
by the isometry
\[
\Vert\psi\Vert(q)=\left(  \sum_{n=0}^{\infty}\frac{1}{n!}\idotsint\Vert
\psi(x_{1},\ldots,x_{n})\Vert^{2}\prod_{i=1}^{n}q(x_{i})\mathrm{d}%
x_{i}\right)  ^{1/2},
\]
where the functions $\psi(x_{1},\ldots x_{n})$ can be extended to the whole of
$\mathbb{X}^{n}$ in a symmetric (or antisymmetric) way uniquely up to the measure zero due to nonatomicity of $\mathrm{d}x$ on $\mathbb{X}$.

\section{Explicit Definition of QS Integrals}

Let $\mathfrak{h}$ be a Hilbert space called the initial space for the Hilbert
products $\mathcal{H}_{\ast}=\mathfrak{h}\otimes\mathcal{K}_{\ast} $ and
$\mathcal{G}_{\ast}=\mathfrak{h}\otimes\mathcal{F}_{\ast}$. We consider the
Hilbert scale $\mathit{G}_{\star}(q)=\mathfrak{h}\otimes\mathit{F}_{\star}(q)$,
$q\in\mathfrak{p}$ of complete tensor products of $\mathfrak{h}$ and the Fock
spaces over $\mathit{K}_{\star}(q)$, and we put%
\[
\mathit{G}_{+}=\cap\mathit{G}_{\star}(q),\;\;\mathit{G}_{-}=\cup\mathit{G}%
_{\star}(q)
\]
which constitute the Gel'fand triple $\mathit{G}_{+}\subseteq\mathcal{G}%
_{\ast}\subseteq\mathit{G}_{-}$ dual to $\mathit{G}^{+}\subseteq
\mathcal{G}\subseteq\mathit{G}^{-}$ of the Hermitian adjoint bra-spaces
$\mathit{G}^{+}=\mathit{G}_{+}^{\ast}$, $\mathcal{G}=\mathcal{G}_{\ast
}^{\ast}$, $\mathit{G}^{-}=\mathit{G}_{-}^{\ast}$.

Let $(\mathrm{D}_{\nu}^{\mu})_{\nu=\circ,+}^{\mu=-,\circ}$ be a quadruple of
functions $\mathrm{D}_{\nu}^{\mu}$ on $\mathbb{X}$ with kernel values $\mathrm{D}_{\nu
}^{\mu}\left(  x\right)  :\mathit{G}_{+}\rightarrow\mathit{G}_{-}$ for
$\mathfrak{k}_{x}=\mathbb{C}$, or, if $\mathfrak{k}_{x}\neq\mathbb{C}$, as
continuous operators
\begin{align}
&  \mathrm{D}_{+}^{-}(x):\mathit{G}_{+}\rightarrow\mathit{G}_{-}%
,\;\;\;\;\;\;\ \mathrm{D}_{\circ}^{\circ}(x):\mathfrak{k}_{x}\otimes
\mathit{G}_{+}\rightarrow\mathfrak{k}_{x}\otimes\mathit{G}_{-}%
,\nonumber\label{2onea}\\
&  \mathrm{D}_{+}^{\circ}(x):\mathfrak{k}_{x}\otimes\mathit{G}_{+}%
\rightarrow\mathit{G}_{-},\;\;\;\;\;\;\mathrm{D}_{\circ}^{-}(x):\mathfrak{k}%
_{x}\otimes\mathit{G}_{+}\rightarrow\mathit{G}_{-}.
\end{align}
The continuity means that there is a $q\in\mathfrak{p}_{1}$ such that these
operators are bounded from $\mathit{G}_{\star}(q)\supseteq\mathit{G}_{+}$ to
$\mathit{G}(q)^{\star}\subseteq\mathit{G}_{-}$, where $\mathit{G}(q)^{\star
}=\mathit{G}_{\star}(q^{-1})$. We assume that $\mathrm{D}_{+}^{-}(x)$ is
locally integrable in the sense that
\[
\exists\,q\in\mathfrak{p}_{1}:\Vert\mathrm{D}_{+}^{-}\Vert_{q,t}^{(1)}%
=\int_{\mathbb{X}^{t}}\Vert\mathrm{D}_{+}^{-}(x)\Vert_{q}\mathrm{d}x<\infty
,\quad\forall t<\infty,
\]
where $\mathbb{X}^{t}=\{x\in \mathbb{X}:t(x)<t\}$, and%
\[
\Vert\mathrm{D}\Vert_{q}=\underset{\chi\in G_\star(q)}{\sup}\{\Vert\mathrm{D}\chi\Vert(q^{-1})/\Vert\chi
\Vert(q)\}
\]
is the norm of the continuous operator $\mathrm{D}:\mathit{G}_{\star
}(q)\rightarrow\mathit{G}_{\star}(q^{-1})$ which defines a bounded Hermitian
form $\mathrm{D}\left(  \chi;\chi\right)  :=\left\langle \chi\mid
\mathrm{D}\chi\right\rangle $ on $\mathit{G}_{\star}(q)$. We also assume that
$\mathrm{D}_{\circ}^{\circ}(x)$ is locally bounded with respect to a strictly
positive function $s$ of $x$ such that $1/s\in\mathfrak{p}_{0}$ in the sense
that
\[
\exists\,q\in\mathfrak{p}_{1}:\Vert\mathrm{D}_{\circ}^{\circ}\Vert
_{q,t}^{(\infty)}(s)=\mathrm{ess}\sup_{x\in \mathbb{X}^{t}}\{s(x)\Vert\mathrm{D}%
_{\circ}^{\circ}(x)\Vert_{q}\}<\infty,\quad\forall t<\infty;
\]
here $\Vert\mathrm{D}\Vert_{q}$ is the norm of the operator $\mathrm{D}%
:\mathfrak{k}_{x}\otimes\mathit{G}_{\star}(q)\rightarrow\mathfrak{k}_{x}%
\otimes\mathit{G}_{\star}(q^{-1})$. Finally, we assume that $\mathrm{D}%
_{+}^{\circ}(x)$ and $\mathrm{D}_{\circ}^{-}(x)$ are locally square integrable
with strictly positive function $r(x)$ such that $1/r\in\mathfrak{p}_{0}$, in
the sense that
\[
\exists\,q\in\mathfrak{p}_{1}:\Vert\mathrm{D}_{+}^{\circ}\Vert_{q,t}%
^{(2)}(r)<\infty,\;\;\;\;\Vert\mathrm{D}_{\circ}^{-}\Vert_{q,t}^{(2)}%
(r)<\infty,\quad\forall t<\infty,
\]
where $\Vert\mathrm{D}\Vert_{q,t}^{(2)}(r)=(\int_{\mathbb{X}^{t}}\Vert\mathrm{D}%
(x)\Vert_{q}^{2}r(x)\mathrm{d}x)^{1/2}$ and $\Vert\mathrm{D}\Vert_{q}$ are the
respective norms of the operators
\[
\mathrm{D}_{+}^{\circ}(x):\mathit{G}_{\star}(q)\rightarrow\mathfrak{k}%
_{x}\otimes\mathit{G}_{\star}(q^{-1}),\;\;\;\;\;\;\,\mathrm{D}_{\circ}%
^{-}(x):\mathfrak{k}_{x}\otimes\mathit{G}_{\star}(q)\rightarrow\mathit{G}%
_{\star}(q^{-1}).
\]
Then for any $t\in\mathbb{R}_{+}$ we can define a \emph{generalized quantum
stochastic (QS) integral}
\begin{equation}
\boldsymbol{i}_{0}^{t}(\mathbf{D})=\int_{\mathbb{X}^{t}}\Lambda(\mathbf{D}%
,\mathrm{d}x),\;\;\;\;\;\Lambda(\mathbf{D},\boldsymbol{\bigtriangleup}%
)=\sum_{\mu,\nu}\Lambda_{\nu}^{\mu}(\mathrm{D}_{\nu}^{\mu}%
,\boldsymbol{\bigtriangleup})\label{2oneb}%
\end{equation}
introduced in {\cite{ParS86}} as the sum of four continuous operators $\Lambda_{\mu
}^{\nu}(\mathrm{D}_{\nu}^{\mu}):\mathit{G}_{+}\rightarrow\mathit{G}_{-}$
described as operator measures on $\mathfrak{F}_{\mathbb{X}}\ni
\boldsymbol{\bigtriangleup}$ for $\boldsymbol{\bigtriangleup}=\mathbb{X}^{t}$ with
values%
\begin{align}
\lbrack\Lambda_{-}^{+}(\mathrm{D}_{+}^{-},\boldsymbol{\bigtriangleup}%
)\chi](\vartheta) &  =\int_{\boldsymbol{\bigtriangleup}}[\mathrm{D}_{+}%
^{-}(x)\chi](\vartheta)\mathrm{d}x\quad\;\text{ \ \ ~(preservation)},\nonumber\\
\lbrack\Lambda_{\circ}^{+}(\mathrm{D}_{+}^{\circ},\boldsymbol{\bigtriangleup
})\chi](\vartheta) &  =\sum_{x\in\boldsymbol{\bigtriangleup}\cap\vartheta
}[\mathrm{D}_{+}^{\circ}(x)\chi](\vartheta\boldsymbol{\setminus}x)\quad\;\;\text{
\ \ (creation)},\nonumber\\
\lbrack\Lambda_{-}^{\circ}(\mathrm{D}_{\circ}^{-},\boldsymbol{\bigtriangleup
})\chi](\vartheta) &  =\int_{\boldsymbol{\bigtriangleup}}[\mathrm{D}_{\circ
}^{-}(x)\mathring{\chi}(x)](\vartheta)\mathrm{d}x\;\;\;\;\text{ (annihilation)}%
,\nonumber\\
\lbrack\Lambda_{\circ}^{\circ}(\mathrm{D}_{\circ}^{\circ}%
,\boldsymbol{\bigtriangleup})\chi](\vartheta) &  =\sum_{x\in
\boldsymbol{\bigtriangleup}\cap\vartheta}[\mathrm{D}_{\circ}^{\circ
}(x)\mathring{\chi}(x)](\vartheta\boldsymbol{\setminus}x)\;\;\;\text{
(exchange)}.\label{2onec}%
\end{align}
Here $\chi\in\mathit{G}_{+},\vartheta\boldsymbol{\setminus}x=\{x^{\prime
}\in\vartheta:x^{\prime}\neq x\}$ denotes the chain $\vartheta\in\mathcal{X}$
from which the point $x\in\vartheta$ has been eliminated, and $\mathring{\chi
}\left(  x\right)  \in\mathfrak{k}_{x}\otimes\mathit{G}_{+}$ is the
\emph{single point split} $\mathring{\chi}\left(  x\right)  =\nabla_{x}\chi$,
or point derivative, defined for each $\chi\in\mathit{G}_{+}$ almost everywhere
(namely, for $\vartheta\in\mathcal{X}$: $x\notin\vartheta$) as the function%
\[
\lbrack\nabla_{x}\chi](\vartheta)=\chi(\vartheta\sqcup x)\equiv\mathring{\chi
}(x,\vartheta),
\]
where the operation $\varkappa\sqcup x$ denotes the disjoint union
$\vartheta=\varkappa\cup x$, $\varkappa\cap x=\emptyset$ of chains
$\varkappa\in\mathcal{X}$ and $x\in \mathbb{X}\boldsymbol{\setminus}\varkappa$
with pairwise comparable elements. Note that the point splitter $\nabla
$ represents the Malliavin derivative {\cite{Mal78}} densely defined in
Fock-Guichardet space as the bosonic annihilation operator$\;b\left(
x\right)  :\mathit{G}_{+}\rightarrow\mathfrak{k}_{x}\otimes\mathit{G}_{+}$ by
$\left[  b(x)\chi\right]  (\vartheta)=\mathring{\chi}(x,\vartheta)$ where one
can take $\mathring{\chi}(x,\vartheta)=0$ if $x\in\vartheta$, and its
inverse operator at $x$ is $\left[  \nabla_{x}^{\ast}\psi\right]  \left(
\vartheta\right)  =\psi\left(  x,\vartheta\boldsymbol{\setminus}x\right)
$ with $\left[  \nabla_{x}^{\ast}\psi\right]  \left(  \vartheta\right)  =0$
if $x\notin\vartheta$ defines in this representation the Skorokhod non-adapted
integral as the creation point integral
\[
\left[  \nabla^{\ast}\psi\right]  \left(  \vartheta\right)
=\sum_{x\in\vartheta}\psi(x,\vartheta\boldsymbol{\setminus}x)
\]
for any $\psi\in\mathit{K}_{-}\otimes\mathit{G}_{-}$. The continuity of this
derivative as the projective limit map $\mathit{G}_{+}\rightarrow
\mathit{K}_{+}\otimes\mathit{G}_{+}$ and the point integral as the adjoint map
$\mathit{K}_{-}\otimes\mathit{G}_{-}\rightarrow\mathit{G}_{-}$ will simply
follow from the isometricity of the multiple point splitter and co-isometricity
of the adjoint multiple point integral as defined below, originally established in {\cite{Be91}}.

\section{Split Operator and its Properties}

As it is proved below, we can consider the multiple bosonic annihilation operators $b^{\otimes
}(\upsilon):\chi\mapsto\mathring{\chi}(\upsilon)$ eliminating several points
$\upsilon$ in $\mathcal{X}_n$, with $\left[  b^{\otimes}(\upsilon)\chi\right]
(\vartheta)=0$ if $\upsilon\subseteq\vartheta$, as partial isometries on the projective limit $G_+$ into $G_+^{(n)}=K_+^{\otimes n}\otimes G_+$. They are described for each
$\upsilon=\{x_{1},\ldots x_{n}\}$ in terms of the $n$\emph{-point split}%
\begin{equation}
\left[  \Delta_{\upsilon}^{\left(  n\right)  }\chi\right]  (\vartheta
):=\chi(\vartheta\sqcup\upsilon)\equiv\chi^{\left(  n\right)  }(\upsilon
,\vartheta),\;\upsilon\in\mathcal{X}_{n},\label{2oneda}%
\end{equation}
where we denoted $\mathring{\chi}\left(  \upsilon\right)  =\chi^{\left(
n\right)  }\left(  \upsilon\right)  $ for $n=\left\vert \upsilon\right\vert $.
It is defined almost everywhere ($\vartheta\cap\upsilon=\emptyset$) on
$\vartheta\in\mathcal{X}$ as the $n$-th order $\Delta_{\upsilon_{n}}%
=\otimes_{x\in\upsilon_{n}}\nabla_{x}$ point (or Malliavin) derivative
{\cite{Mal78}} such that $\Delta^{\left(  0\right)  }=\mathrm{I}$ and
$\Delta^{\left(  1\right)  }=\nabla$. These $n$-tuple annihilations, densely defined
as operators from $\mathit{G}_{+}$ into $\mathfrak{k}^{\otimes}\left(
\upsilon\right)  \otimes\mathit{G}_{+}$, are not continuous for each
$\upsilon\in\mathcal{X}_{n}$ (except $\upsilon=\emptyset$ corresponding to
$n=0$ for which $b^{\otimes}\left(  \emptyset\right)  =\mathrm{I}$), but they
define projective-continuous linear maps into the space $\mathit{G}%
_{+}^{\left(  n\right)  }$ of
functions
$\upsilon\mapsto\psi\left(  \upsilon\right)  $ on $\mathcal{X}%
_{n}\subset\mathcal{X}$ for each $n\in\mathbb{N}$, and therefore have the adjoints $G_+^{(n)}\rightarrow G_+$. This follows from the projective contractivity of the maps
$\Delta^{\left(  n\right)  }:$ $\mathit{G}_{+}\rightarrow\mathit{G}%
_{+}^{\left(  n\right)  }$, and their adjoints, such that each function
$\chi^{\left(  n\right)  }=\Delta^{\left(  n\right)  }\chi$ is
square-integrable on $\mathcal{X}_{n}$ with any $q_{0}\in\mathfrak{p}_{0}$
being a component of the isometric operator%
\[
\Delta\chi=\int_{\mathcal{X}}^{\oplus}\Delta_{\upsilon}\chi\mathrm{d}%
\upsilon=\oplus_{n=0}^{\infty}\Delta^{\left(  n\right)  }\chi\equiv
\mathring{\chi}.%
\]
The projective isometricity of the linear operator $\Delta=\int_{\mathcal{X}%
}^{\oplus}\Delta_{\upsilon}\mathrm{d}\upsilon$, called the \emph{multiple point splitter},
\[
\Delta:\mathit{G}_{\star}\left(  q_{0}+q_{1}\right)  \rightarrow\mathit{F}%
_{\star}\left(  q_{0}\right)  \otimes\mathit{G}_{\star}(q_{1})
\]
is established in the following lemma.

\begin{lemma}
The linear map $\Delta:\chi\mapsto\left[  \Delta^{\left(  n\right)  }%
\chi\right]  $ defined as $\Delta\chi=\oplus_{n=0}^{\infty}\Delta^{\left(
n\right)  }\chi$ in (\ref{2oneda}) for all $\upsilon\in\mathcal{X}$ is a
projective isometry on Hilbert scale $\left\{  \mathit{G}_{\star}%
(q):q\in\mathfrak{p}\right\}  $ into the scale $\left\{  \mathit{F}_{\star
}\left(  q_{0}\right)  \otimes\mathit{G}_{\star}(q_{1}):q_{0}\in\mathfrak{p}%
_{0},q_{1}\in\mathfrak{p}_{1}\right\}  $ such that%
\[
\Vert\Delta\chi\Vert\left(  q_{0},q_{1}\right)  =\Vert\chi\Vert\left(
q_{0}+q_{1}\right)  .
\]
The adjoint co-isometric operator $\left\langle \Delta^{\ast}\psi
|\chi\right\rangle =\left\langle \psi|\Delta\chi\right\rangle $, defined on
$\psi\in\mathit{F}_{\star}\left(  q_{0}\right)  \otimes\mathit{G}_{\star}%
(q_{1})$ as the multiple point integral $\Delta^{\ast}=\sum\Delta
_{n}^{\ast}$, is a contraction from $\mathit{F}_{\star}(q_{0}^{-1}%
)\otimes\mathit{G}_{\star}\left(  q_{1}^{-1}\right)  $ into any $\mathit{G}%
_{\star}\left(  q^{-1}\right)  $ with $q\geq q_{0}+q_{1}$ such that $\left(
\Delta_{n}^{\ast}\right)  ^{\ast}=\Delta^{\left(  n\right)  }$. In
particular,%
\begin{equation}
\lbrack\Delta_{n}^{\ast}\psi](\vartheta)=\sum_{\mathcal{X}_{n}\ni
\upsilon\subseteq\vartheta}\psi(\upsilon,\vartheta\boldsymbol{\setminus
}\upsilon),\;\;\;\vartheta\in\mathcal{X}\label{2onedb}%
\end{equation}
defines for $\psi\left(  \upsilon,\varkappa\right)  =\psi_{n}\left(
\upsilon\right)  \otimes\chi\left(  \varkappa\right)  $ the $n$-th order
Skorokhod integral%
\[
\left[  S_{n}\left(  \psi_{n}\right)  \chi\right]  \left(  \vartheta\right)
=\sum_{\mathcal{X}_{n}\ni\upsilon\subseteq\vartheta}\psi_{n}(\upsilon
)\otimes\chi(\vartheta\boldsymbol{\setminus}\upsilon)=\left[\Delta_{n}%
^{\ast}\left(  \psi_{n}\otimes\chi\right)\right](\vartheta)
\]
of $\psi_{n}\in\mathit{K}_{\star}^{(n)}\left(  q_{0}%
^{-1}\right)  $ on $\chi\in\mathit{G}_{\star}\left(  q_{1}^{-1}\right)  $.
\end{lemma}

\begin{proof}
We first of all establish the principal formula of the multiple integration
\begin{equation}
\int\sum_{\upsilon\subseteq\vartheta}f(\upsilon,\vartheta
\boldsymbol{\setminus}\upsilon)\mathrm{d}\vartheta=\iint f(\upsilon
,\varkappa)\mathrm{d}\upsilon\mathrm{d}\varkappa,\quad\forall f\in
\mathit{L}^{1}(\mathcal{X}\times\mathcal{X}),\label{2oned}%
\end{equation}
which will allow us to define the adjoint operator $\Delta^{\ast}$. Let
$f(\upsilon,\varkappa)=g(\upsilon)h(\varkappa)$ be the product of integrable
complex functions on $\mathcal{X}$ of the form $g(\upsilon)=\prod
_{x\in\upsilon}g(x)$, $h(\varkappa)=\prod_{x\in\varkappa}h(x)$ for any
$\upsilon$, $\varkappa\in\mathcal{X}$. Employing the binomial formula
\[
\sum_{\upsilon\subseteq\vartheta}g(\upsilon)h(\vartheta
\boldsymbol{\setminus}\upsilon)=\sum_{\upsilon\sqcup\varkappa=\vartheta
}\prod_{x\in\upsilon}g(x)\prod_{x\in\varkappa}h(x)=\prod_{x\in\vartheta
}(g(x)+h(x)),
\]
and also the equality $\int f(\upsilon)\mathrm{d}\upsilon=\exp\{\int
f(x)\mathrm{d}x\}$ for $f(\upsilon)=\prod_{x\in\upsilon}f(x)$, we obtain the
formula
\[
\int\sum_{\upsilon\subseteq\vartheta}g(\upsilon)h(\vartheta
\boldsymbol{\setminus}\upsilon)\mathrm{d}\vartheta=\exp\left\{
\int(g(x)+h(x))\mathrm{d}x\right\}  =\iint g(\upsilon)h(\varkappa
)\mathrm{d}\upsilon\mathrm{d}\varkappa,
\]
which proves~(\ref{2oned}) on a set of product-functions $f$ dense in
$\mathit{L}^{1}(\mathcal{X}\times\mathcal{X})$.

Applying this formula to the scalar product $\left\langle \psi(\upsilon
,\varkappa)|\psi(\upsilon,\varkappa)\right\rangle \in\mathit{L}^{1}%
(\mathcal{X}\times\mathcal{X})$, we obtain
\[
\int\sum_{\upsilon\subseteq\vartheta}\left\langle \psi(\upsilon,\vartheta
\boldsymbol{\setminus}\upsilon)\mid\chi(\vartheta)\right\rangle
\mathrm{d}\vartheta=\iint\left\langle \psi(\upsilon,\varkappa)\mid
\chi(\upsilon\sqcup\varkappa)\right\rangle \mathrm{d}\upsilon\mathrm{d}%
\varkappa,
\]
that is, $\left\langle \Delta^{\ast}\psi\mid\chi\right\rangle =\left\langle
\psi\mid\Delta\chi\right\rangle $, where $[\Delta\chi](\upsilon,\varkappa
)=\chi(\varkappa\sqcup\upsilon)\equiv\mathring{\chi}(\upsilon,\varkappa)$.
Choosing arbitrary $\psi\in\mathit{F}_{\star}(q_{0}^{-1})\otimes\mathit{G}%
_{\star}\left(  q_{1}^{-1}\right)  $, we find that the annihilation operators
$b(\upsilon)\chi=[\Delta_{\upsilon}\chi] $ define the isometry $\Delta
:\mathit{G}_{\star}\left(  q_{0}+q_{1}\right)  \rightarrow\mathit{F}_{\star
}\left(  q_{0}\right)  \otimes\mathit{G}_{\star}(q_{1})$ with the operator
$\Delta^{\ast}$ defined as co-isometry $\mathit{F}_{\star}(q_{0}^{-1}%
)\otimes\mathit{G}_{\star}\left(  q_{1}^{-1}\right)  \rightarrow\mathit{G}%
_{\star}\left(  q^{-1}\right)  $ for$\,q=q_{0}+q_{1}$ with respect to the
standard pairing of dual spaces $\mathit{G}_{\star}(q)$ and $\mathit{G}_{\star
}\left(  q^{-1}\right)  $:
\begin{align*}
&  \left\Vert \mathring{\chi}\right\Vert ^{2}\left(  q_{0},q_{1}\right)
=\iint\Vert\mathring{\chi}\left(  \upsilon,\varkappa\right)  \Vert^{2}%
q_{0}(\upsilon)q_{1}(\varkappa)\mathrm{d}\upsilon\mathrm{d}\varkappa\\
&  =\int\sum_{\upsilon\subseteq\vartheta}\Vert\chi(\vartheta)\Vert^{2}%
q_{0}(\upsilon)q_{1}(\vartheta\boldsymbol{\setminus}\upsilon
)\mathrm{d}\vartheta=\int\Vert\chi(\vartheta)\Vert^{2}\sum_{\upsilon
\sqcup\varkappa=\vartheta}q_{0}(\upsilon)q_{1}(\varkappa)\mathrm{d}\vartheta\\
&  =\int\Vert\chi(\vartheta)\Vert^{2}(q_{0}+q_{1})(\vartheta)\mathrm{d}%
\vartheta\equiv\left\Vert \chi\right\Vert^2 \left(  q_{0}+q_{1}\right) \leq\|\chi\|^2(q)\;\;\forall\;q\geq q_0+q_1.
\end{align*}
Hence it follows that $\Delta$ is projective continuous operator from
$\mathit{G}_{+}$ to $\mathit{F}_{+}\otimes\mathit{G}_{+}$, where
$\mathit{F}_{+}=\bigcap_{q\in\mathfrak{p}_{0}}\mathit{F}_{\star}(q)$, and in
particular so is the one-point split $\chi(x,\varkappa)=\chi(x\sqcup
\varkappa)\equiv\mathring{\chi}\left(  x,\varkappa\right)  $ from
$\mathit{G}_{+}$ to $\mathit{K}_{+}\otimes\mathit{G}_{+}$, as a contracting map
$\mathit{G}_{\star}\left(  q_{0}+q_{1}\right)  \rightarrow\mathit{F}_{\star
}\left(  q_{0}\right)  \otimes\mathit{G}_{\star}(q_{1})$ for all $q_{0}%
\in\mathfrak{p}_{0},q_{1}\in\mathfrak{p}$. The lemma is proved.
\end{proof}
\begin{xrem}
Because the explicit form of both the creation and annihilation operators' norms may not be obvious we shall review them here for the reader's familiarization.
\[
\|\Delta\|^{}_{q_0+q_1}=\sup_\chi\frac{\|\Delta\chi\|(q_0,q_1)}{\|\chi\|(q_0+q_1)}= \sup_{\chi,\psi}\frac{|\langle\psi|\Delta\chi\rangle|}{\|\psi\|(q_0^{-1},q_1^{-1})\|\chi\|(q_0+q_1)}
\]
and
\[
\|\Delta^\ast\|^{q_0,q_1}_{}=\sup_\psi\frac{\|\Delta^\ast\psi\|((q_0+q_1)^{-1})}{\|\psi\|(q_0^{-1},q_1^{-1})}= \sup_{\chi,\psi}\frac{|\langle\Delta^\ast\psi|\chi\rangle|}{\|\psi\|(q_0^{-1},q_1^{-1})\|\chi\|(q_0+q_1)}
\]
where $\psi\in F_\star(q_0^{-1})\otimes G_\star(q_1^{-1})$ and $\chi\in G_\star(q_0+q_1)$, indeed $\Delta^\ast:F_\star(q_0^{-1})\otimes G_\star(q_1^{-1})\rightarrow G_\star((q_0+q_1)^{-1})$ and $\Delta:G_\star(q_0+q_1)\rightarrow F_\star(q_0)\otimes G_\star(q_1)$. By virtue of the fact that $\langle\psi|\Delta\chi\rangle=\langle\Delta^\ast\psi|\chi\rangle$ these two norms are equal, and they are equal to 1.
\end{xrem}

\section{Multiple QS Integrals and Their Continuity}

We are now ready to prove the \emph{inductive continuity} of the integral
(\ref{2oneb}) with respect to $\mathbf{D}=\left[  \mathrm{D}_{\nu}^{\mu
}\right]  $ by showing the inequality
\[
\Vert(\boldsymbol{i}_{0}^{t}(\mathbf{D})\chi)\Vert\left(  p^{-1}\right)
\leq\Vert\mathrm{D}\Vert_{q,t}^{s}(r)\Vert\chi\Vert(p),\quad\forall p\geq
r^{-1}+q+s^{-1},
\]
where $\Vert\mathrm{D}\Vert_{q,t}^{s}(r)=\Vert\mathrm{D}_{+}^{-}\Vert
_{q,t}^{(1)}+\Vert\mathrm{D}_{+}^{\circ}\Vert_{q,t}^{(2)}(r)+\Vert
\mathrm{D}_{\circ}^{-}\Vert_{q,t}^{(2)}(r)+\Vert\mathrm{D}_{\circ}^{\circ
}\Vert_{q,t}^{(\infty)}(s)$. We will establish this inequality as the
single-integral case of the corresponding inequality for the
\emph{\ generalized multiple QS integral} {\cite{Be91}}
\begin{equation}
\lbrack\boldsymbol{\imath}_{0}^{t}(\mathrm{M})\chi](\vartheta)=\sum
_{\upsilon_{\circ}^{\circ}\sqcup\upsilon_{+}^{\circ}\subseteq\vartheta^{t}%
}\int_{\mathcal{X}^{t}}\int_{\mathcal{X}^{t}}[\mathrm{M}(\boldsymbol{\upsilon
})\mathring{\chi}(\upsilon_{\circ}^{-}\sqcup\upsilon_{\circ}^{\circ
})](\upsilon_{-}^{\circ})\mathrm{d}\upsilon_{+}^{-}\mathrm{d}\upsilon_{\circ
}^{-}\label{2onee}%
\end{equation}
where $\vartheta^{t}=\vartheta\cap \mathbb{X}^{t},\mathcal{X}^{t}=\{\vartheta
\in\mathcal{X}:\vartheta\subset \mathbb{X}^{t}\}$ and the sum is taken over all
decompositions $\vartheta=\upsilon_{-}^{\circ}\sqcup\upsilon_{\circ}^{\circ
}\sqcup\upsilon_{+}^{\circ}$ such that $\upsilon_{\circ}^{\circ}\in
\mathcal{X}^{t}$ and $\upsilon_{+}^{\circ}\in\mathcal{X}^{t} $. The
\emph{multi-integrand }$\mathrm{M}\left(  \boldsymbol{\upsilon}\right)  $
defines the values%
\[
\mathrm{M}\left(  \upsilon_{\cdot}^{\cdot}\right)  =\delta_{\emptyset}\left(
\upsilon_{\cdot}^{+}\right)  \mathrm{M}\left(  \boldsymbol{\upsilon}\right)
\delta_{\emptyset}\left(  \upsilon_{-}^{\cdot}\right)
,\;\;\boldsymbol{\upsilon}=\left(
\begin{array}
[c]{cc}%
\upsilon_{+}^{-} & \upsilon_{\circ}^{-}\\
\upsilon_{+}^{\circ} & \upsilon_{\circ}^{\circ}%
\end{array}
\right)
\]
of matrix elements $\mathrm{M}^\cdot_\cdot(\upsilon)=\mathrm{M}(\upsilon^\cdot_\cdot)$, $\upsilon=\sqcup_{\mu,\nu}\upsilon^\mu_\nu$, for a decomposable triangular tensor-operator $\mathbf{M}\equiv\mathrm{M}^\cdot_\cdot $ defined on
matrices $\upsilon_{\cdot}^{\cdot}=\left[  \upsilon_{\nu}^{\mu}\right]
_{\nu=-,\circ,+}^{\mu=-,\circ,+}$ whose elements are finite disjoint chains
$\upsilon_{\nu}^{\mu}$, with $\upsilon^\mu_\nu=\emptyset$ for $\mu>\nu$ and with $\mathrm{M}(\upsilon^\cdot_\cdot)=0$ if $\upsilon^+_\cdot\neq\emptyset$ or $\upsilon^\cdot_-\neq\emptyset$. The other values of $\mathrm{M}\left(  \upsilon_{\cdot}^{\cdot
}\right)$ are defined by a kernel-operator function $\mathrm{M}\left(  \boldsymbol{\upsilon}\right)  $ of the quadruple $\boldsymbol{\upsilon}%
=(\upsilon_{\nu}^{\mu})_{\nu=+,\circ}^{\mu=-,\circ}$ with
continuous operator values $\mathrm{M}\left(  \boldsymbol{\upsilon}\right)
:\mathit{G}_{+}\rightarrow\mathit{G}_{-}$ in the scalar case $\mathfrak{k}_x%
=\mathbb{C}$. In the general case it is defined almost everywhere by its
values on the chains $\upsilon_{\nu}^{\mu}\in\mathcal{X}$ in the continuous
operators
\[
\mathrm{M}\left(
\begin{array}
[c]{cc}%
\upsilon_{+}^{-} & \upsilon_{\circ}^{-}\\
\upsilon_{+}^{\circ} & \upsilon_{\circ}^{\circ}%
\end{array}
\right)  :\mathfrak{k}_{\upsilon_{\circ}^{-}}^{\otimes}\otimes\mathfrak{k}%
_{\upsilon_{\circ}^{\circ}}^{\otimes}\otimes\mathit{G}_{+}\rightarrow
\mathfrak{k}_{\upsilon_{\circ}^{\circ}}^{\otimes}\otimes\mathfrak{k}%
_{\upsilon_{+}^{\circ}}^{\otimes}\otimes\mathit{G}_{-}.
\]
We will assume that these operators are bounded from $\mathit{G}_{\star}(q)$ to
$\mathit{G}_{\star}\left(  q^{-1}\right)  $ for some $q\in\mathfrak{p}_{1}$, such that $\|{\mathrm{I}}\|_q=1$, and
that there exist strictly positive functions $r>0$, $r^{-1}\in\mathfrak{p}%
_{0}$, and $s>0$, $s^{-1}\in\mathfrak{p}_{0}$ such that
\begin{equation}
\Vert\mathrm{M}\Vert_{q,t}^{s}(r)=\int_{\mathcal{X}^{t}}\Vert\mathrm{M}%
_{+}^{-}(\upsilon)\Vert_{q,t}^{s}\left(  r\right)  \mathrm{d}\upsilon
<\infty,\quad\forall t<\infty,\label{2onef}%
\end{equation}
where
\[
\Vert\mathrm{M}_{+}^{-}(\upsilon_{+}^{-})\Vert_{q,t}^{s}(r)=\left(
\int_{\mathcal{X}^{t}}\int_{\mathcal{X}^{t}}\mathrm{ess}\sup_{\upsilon_{\circ
}^{\circ}\in\mathcal{X}^{t}}(s(\upsilon_{\circ}^{\circ})\Vert\mathrm{M}%
(\boldsymbol{\upsilon})\Vert_{q})^{2}r(\upsilon_{+}^{\circ}\sqcup
\upsilon_{\circ}^{-})\mathrm{d}\upsilon_{+}^{\circ}\mathrm{d}\upsilon_{\circ
}^{-}\right)  ^{\frac{1}{2}},
\]
and $s(\upsilon)=\prod_{x\in\upsilon}s(x)$, $r(\upsilon)=\prod_{x\in\upsilon
}r(x)$.

We mention that the single integral~(\ref{2oneb}) corresponds to the
case
\[
\mathrm{M}(\upsilon_{\cdot}^{\cdot})=0,\quad\forall\upsilon_{\cdot}^{\cdot
}:\sum_{\mu\neq+,\nu\neq-}|\upsilon_{\nu}^{\mu}|\neq1,
\]
and $\mathrm{M}(\mathbf{x}_{\nu}^{\mu})=\mathrm{D}_{\nu}^{\mu}(x)$ otherwise, where $\mathbf{x}_{\nu}^{\mu}$ denotes one of six `atomic' triangular matrices
$\upsilon_{\cdot}^{\cdot}(\mathbf{x})=\left[  \upsilon_{\lambda}^{\kappa}(\mathbf{x})\right]
_{\lambda=-,\circ,+}^{\kappa=-,\circ,+}\equiv\mathbf{x}$ having all matrix elements
$\upsilon_{\lambda}^{\kappa}(\mathbf{x})$ empty if $\mathbf{x}\neq\mathbf{x}^\kappa_\lambda$, but $\upsilon_{\nu}^{\mu}(\mathbf{x})=x$ for $\mathbf{x}=\mathbf{x}^\mu_\nu$. Note that
integrand $\mathrm{M}(\mathbf{x}_{\nu}^{\mu})$ is zero on the atomic matrices
$\mathbf{x}_{\nu}^{+}$ and $\mathbf{x}_{-}^{\mu}$, otherwise $\mathrm{M}%
(\mathbf{x})=\mathrm{M}(\boldsymbol{x}) $, given by the single-point kernel $\mathrm{M}(\boldsymbol{x})$ as a function of one of the four single-point tables
$\boldsymbol{\upsilon}\left(  x\right)= \left(\upsilon^\kappa_\lambda(\mathbf{x})\right)^{\kappa=-,\circ}_{\lambda=+,\circ} \equiv\boldsymbol{x}$:
\begin{equation}
\boldsymbol{x}_{+}^{-}=%
\begin{pmatrix}
x, & \emptyset\\
\emptyset & \emptyset
\end{pmatrix}
,\boldsymbol{x}_{+}^{\circ}=%
\begin{pmatrix}
\emptyset, & \emptyset\\
x & \emptyset
\end{pmatrix}
,\boldsymbol{x}_{\circ}^{-}=%
\begin{pmatrix}
\emptyset, & x\\
\emptyset & \emptyset
\end{pmatrix}
,\boldsymbol{x}_{\circ}^{\circ}=%
\begin{pmatrix}
\emptyset, & \emptyset\\
\emptyset & x
\end{pmatrix}
,\label{2oneg}%
\end{equation}
determined by an $x\in \mathbb{X}$. It follows from the next theorem that the function
$\mathrm{M}(\boldsymbol{\upsilon})$ in~(\ref{2onee}) can be defined up to
equivalence, whose kernel $\left\{  \mathrm{M}\approx0\right\}  $ consists of
all multiple integrands with $\Vert\mathrm{M}\Vert_{q,t}^{s}(r)=0$ for all
$t\in\mathbb{R}_{+}$ and for some $q,r,s$. In particular, $\mathrm{M}$ can be
defined almost everywhere only for the tables $\boldsymbol{\upsilon}%
=(\upsilon_{\nu}^{\mu})$ that give disjoint decompositions $\upsilon
=\sqcup_{\mu,\nu}\upsilon_{\nu}^{\mu}$ of the chains $\upsilon\in\mathcal{X}$,
that is, it may have nonzero values only on $\boldsymbol{\upsilon}$
representable in the form $\boldsymbol{\upsilon}=\sqcup_{x\in\upsilon
}\boldsymbol{x}$, where $\boldsymbol{x}$ is one of the atomic tables given
by~(\ref{2oneg}) with indices $\mu,\nu$ for $x\in\upsilon_{\nu}^{\mu}$.

\begin{theorem}
Suppose that $\mathrm{M}(\boldsymbol{\upsilon})$ is a locally
integrable function in the sense of \textup{(\ref{2onef})} for some $q,r,s>0$. Then its
integral \textup{(\ref{2onee})} is a continuous operator $\mathrm{T}%
_{t}=\boldsymbol{\imath}_{0}^{t}(\mathrm{M})$ from $\mathit{G}_{+}$ to
$\mathit{G}_{-}$ satisfying the estimate
\begin{equation}
\Vert\mathrm{T}_{t}\Vert_{p}=\sup_{\chi\in\mathit{G}_{\star}(p)}\left\{
\Vert\mathrm{T}_{t}\chi\Vert\left(  p^{-1}\right)  /\Vert\chi\Vert(p)\right\}
\leq\Vert\mathrm{M}\Vert_{q,t}^{s}(r)\label{2oneh}%
\end{equation}
for any $p\geq r^{-1}+q+s^{-1}$. The operator $\mathrm{T}_{t}^{\ast}$,
formally adjoint to $\mathrm{T}_{t}$ in $\mathcal{G}_{\ast}$, is the integral
\begin{equation}
\boldsymbol{\imath}_{0}^{t}(\mathrm{M})^{\ast}=\boldsymbol{\imath}_{0}%
^{t}(\mathrm{M}^{\ddagger}),\quad\mathrm{M}^{\ddagger}\left(  \upsilon_{\cdot
}^{\cdot}\right)  =\mathrm{M}\left(  \widetilde{\upsilon_{\cdot}^{\cdot}%
}\right)  ^{\ast},\;\widetilde{\left[  \upsilon_{\nu}^{\mu}\right]
}=\left[  \upsilon_{-\mu}^{-\nu}\right] \label{2onei}%
\end{equation}
of $\delta_{\emptyset}\left(  \upsilon_{\cdot}^{+}\right)  \mathrm{M}^{\star
}\left(  \boldsymbol{\upsilon}\right)  \delta_{\emptyset}\left(  \upsilon
_{-}^{\cdot}\right)  =\mathrm{M}^{\ddagger}\left(  \upsilon_{\cdot}^{\cdot
}\right)  $, which is continuous from $\mathit{G}_{+}$ to $\mathit{G}_{-}$,
and satisfying $\Vert\mathrm{M}^{\star}\Vert_{q,t}^{s}(r)=\Vert\mathrm{M}\Vert_{q,t}%
^{s}(r)$ for $\mathrm{M}^{\star}\left(  \boldsymbol{\upsilon}\right)
=\mathrm{M}\left(  \boldsymbol{\upsilon}^{\prime}\right)  ^{\ast} $, where
$\left(  \upsilon_{\nu}^{\mu}\right)  ^{\prime}=\left(  \upsilon_{-\mu}^{-\nu
}\right)  $. Moreover, the operator-valued function $t\mapsto\mathrm{T}_{t}$
has the quantum-stochastic differential $\mathrm{dT}_{t}=\mathrm{d}%
\boldsymbol{i}_{0}^{t}(\mathbf{D})$ in the sense that
\begin{equation}
\boldsymbol{\imath}_{0}^{t}(\mathrm{M})=\mathrm{M}(\emptyset)+\boldsymbol{i}%
_{0}^{t}(\mathbf{D}),\quad\mathrm{D}_{\nu}^{\mu}(x)=\boldsymbol{\imath}%
_{0}^{t(x)}(\mathrm{\dot{M}}(\mathbf{x}_{\nu}^{\mu})),\label{2onej}%
\end{equation}
defined by the quantum-stochastic derivatives $\mathbf{D}=\left[
\mathrm{D}_{\nu}^{\mu}\right]  $ with values \textup{(\ref{2onea})} acting
from $\mathit{G}_{\star}(p)$ to $\mathit{G}_{\star}\left(  p^{-1}\right)  $ and
bounded almost everywhere:
\[
\Vert\mathrm{D}_{+}^{-}\Vert_{p,t}^{(1)}\leq\Vert\mathrm{M}\Vert_{q,t}%
^{s}(r),\quad\Vert\mathrm{D}\Vert_{p,t}^{(2)}(r)\leq\Vert\mathrm{M}\Vert
_{q,t}^{s}(r),\quad\Vert\mathrm{D}_{\circ}^{\circ}\Vert_{p,t}^{(\infty
)}(s)\leq\Vert\mathrm{M}\Vert_{q,t}^{s}(r)
\]
for $\mathrm{D}=\mathrm{D}_{\circ}^{-}$ and $\mathrm{D}=\mathrm{D}_{+}^{\circ
}$, $p\geq r^{-1}+q+s^{-1}$. This differential is defined in the form of the
multiple integrals \textup{(\ref{2onee})}, with respect to
$\boldsymbol{\upsilon}$, of the point derivatives $\mathrm{\dot{M}}%
(\mathbf{x},\upsilon_{\cdot}^{\cdot})=\mathrm{M}(\upsilon_{\cdot}^{\cdot
}\sqcup\mathbf{x})$, where $\mathbf{x}$ is given by one of four atomic tables
\textup{(\ref{2oneg})} at a fixed point $x\in \mathbb{X}$.
\end{theorem}

\begin{proof}
Using property (\ref{2oned}) in the form
\[
\int\sum_{\sqcup\upsilon_{\nu}^{\circ}=\vartheta}f(\upsilon_{-}^{\circ
},\upsilon_{\circ}^{\circ},\upsilon_{+}^{\circ})\mathrm{d}\vartheta=\iiint
f(\upsilon_{-}^{\circ},\upsilon_{\circ}^{\circ},\upsilon_{+}^{\circ}%
)\prod_{\nu}\mathrm{d}\upsilon_{\nu}^{\circ},
\]
it is easy to find that from the definition (\ref{2onee}) for $\chi
\in\mathit{G}_{+}$ we have%
\begin{align*}
&{\int\left\langle \chi(\vartheta)\mid\lbrack\mathrm{T}_{t}\chi](\vartheta
)\right\rangle \mathrm{d}\vartheta}\\
&  =\int_{\mathcal{X}^{t}}\mathrm{d}\upsilon_{+}^{-}\int_{\mathcal
{X}^{t}}\mathrm{d}\upsilon_{+}^{\circ}\int_{\mathcal{X}^{t}}\mathrm{d}%
\upsilon_{\circ}^{-}\int_{\mathcal{X}^{t}}\mathrm{d}\upsilon_{\circ}^{\circ
}\left\langle \mathring{\chi}(\upsilon_{\circ}^{\circ}\sqcup\upsilon
_{+}^{\circ})\mid\mathrm{M}(\boldsymbol{\upsilon})\mathring{\chi}(\upsilon_{\circ
}^{-}\sqcup\upsilon_{\circ}^{\circ})\right\rangle \\
&  =\int_{\mathcal{X}^{t}}\mathrm{d}\upsilon_{+}^{-}\int_{\mathcal{X}^{t}%
}\mathrm{d}\upsilon_{+}^{\circ}\int_{\mathcal{X}^{t}}\mathrm{d}\upsilon
_{\circ}^{-}\int_{\mathcal{X}^{t}}\mathrm{d}\upsilon_{\circ}^{\circ
}\left\langle \mathrm{M}(\boldsymbol{\upsilon})^{\ast}\mathring{\chi
}(\upsilon_{\circ}^{\circ}\sqcup\upsilon_{+}^{\circ})\mid\mathring{\chi
}(\upsilon_{\circ}^{-}\sqcup\upsilon_{\circ}^{\circ})\right\rangle \\
&  =\int\left\langle [\mathrm{T}_{t}^{\ast}\chi](\vartheta)\mid
\chi(\vartheta)\right\rangle \mathrm{d}\vartheta,
\end{align*}
that is, $\mathrm{T}_{t}^{\ast}$ acts as $\boldsymbol{\imath}_{0}%
^{t}(\mathrm{M}^{\star})$ in (\ref{2onee}) with $\mathrm{M}^{\star
}(\boldsymbol{\upsilon})=\mathrm{M}(\boldsymbol{\upsilon}^{\prime})^{\ast}%
$, where $(\upsilon_{\nu}^{\mu})^{\prime}=(\upsilon_{-\mu}^{-\nu})$ with
respect to the inversion $-:(-,\circ,+)\mapsto(+,\circ,-)$. More precisely,
this yields $\Vert\boldsymbol{\imath}_{0}^{t}(\mathrm{M})\Vert_{p}%
=\Vert\boldsymbol{\imath}_{0}^{t}(\mathrm{M}^{\star})\Vert_{p}$, since
$\Vert\mathrm{T}\Vert_{p}=\Vert\mathrm{T}^{\ast}\Vert_{p}$ by the
definition (\ref{2oneh}) of $p$-norm and by
\[
\sup\left\{  |\left\langle \chi\mid\mathrm{T}\chi\right\rangle |/\Vert
\chi\Vert(p)\Vert\chi\Vert(p)\right\}  =\sup\{|\left\langle \mathrm{T}%
^{\ast}\chi\mid\chi\right\rangle |/\Vert\chi\Vert(p)\Vert\chi\Vert(p)\}.
\]
We estimate the integral $\left\langle \chi\mid\mathrm{T}_{t}\chi\right\rangle
$ using the Schwartz inequality
\[
\int\Vert\mathring{\chi}(\upsilon)\Vert(q)\Vert\mathring{\chi}(\upsilon
)\Vert(q)s^{-1}(\upsilon)\mathrm{d}\upsilon\leq\Vert\mathring{\chi}%
\Vert(s^{-1},q)\Vert\mathring{\chi}\Vert(s^{-1},q)
\]
and the property (\ref{2oned}) of the multiple integral according to which
\[
\Vert\mathring{\chi}\Vert(s^{-1},q)=\Vert\chi\Vert(q+s^{-1})
\]
then
\begin{align*}
\lefteqn{\left\vert \left\langle \chi\mid\mathrm{T}_{t}\chi\right\rangle
\right\vert \leq
\int_{\mathcal{X}^t}\int_{\mathcal{X}^t}\int_{\mathcal{X}^t}\int_{\mathcal{X}^t} \|\mathrm{\mathring{\chi}}(\upsilon_\circ^\circ\sqcup\upsilon^\circ_+)\|(q) \|\mathrm{M}(\boldsymbol{\upsilon})\mathrm{\mathring{\chi}}(\upsilon^\circ_\circ\sqcup\upsilon^-_\circ)\|(q^{-1}) \mathrm{d}^4\boldsymbol{\upsilon}} \\
& \quad\leq\int_{\mathcal{X}^{t}}\mathrm{d}\upsilon_{\circ
}^{\circ}\int_{\mathcal{X}^{t}}\int_{\mathcal{X}^{t}}\Vert\mathring{\chi
}(\upsilon_{\circ}^{\circ}\sqcup\upsilon_{+}^{\circ})\Vert(q)\left(
\int_{\mathcal{X}^{t}}\Vert\mathrm{M}(\boldsymbol{\upsilon})\Vert_{q}\mathrm
{d}\upsilon_{+}^{-}\right)  \Vert\mathring{\chi}(\upsilon_{\circ}^{-}%
\sqcup\upsilon_{\circ}^{\circ})\Vert(q)\mathrm{d}\upsilon_{\circ}^{-}%
\mathrm{d}\upsilon_{+}^{\circ}\\
&  \quad\leq\int_{\mathcal{X}^{t}}\mathrm{d}\upsilon\left(  \int
_{\mathcal{X}^{t}}\Vert\mathring{\chi}(\upsilon\sqcup\upsilon_{+}^{\circ
})\Vert^{2}(q)\frac{\mathrm{d}\upsilon_{+}^{\circ}}{r(\upsilon_{+}^{\circ}%
)}\int_{\mathcal{X}^{t}}\Vert\mathring{\chi}(\upsilon\sqcup\upsilon_{\circ
}^{-})\Vert^{2}(q)\frac{\mathrm{d}\upsilon_{\circ}^{-}}{r(\upsilon_{\circ}%
^{-})}\right)  ^{\frac{1}{2}}\Vert\mathrm{M}_{\circ}^{\circ}%
(\upsilon)\Vert_{q,t}(r)\\
&  \quad=\int_{\mathcal{X}^{t}}\mathrm{d}\upsilon\Vert\mathring{\chi}%
(\upsilon)\Vert(r^{-1}+q)\Vert\mathrm{M}_{\circ}^{\circ}(\upsilon)\Vert
_{q,t}(r)\Vert\mathring{\chi}(\upsilon)\Vert(r^{-1}+q)\\
&  \quad\leq\underset{\upsilon\in\mathcal{X}^{t}}{\mathrm{ess}\sup
}\Big\{s(\upsilon)\Vert\mathrm{M}_{\circ}^{\circ}(\upsilon)\Vert_{q,t}%
(r)\Big\}\Vert\chi\Vert(r^{-1}+q+s^{-1})\Vert\chi\Vert(r^{-1}+q+s^{-1}),
\end{align*}
where $\Vert\mathrm{M}_{\circ}^{\circ}(\upsilon_{\circ}^{\circ})\Vert
_{q,t}(r)=(\int_{\mathcal{X}^{t}}\int_{\mathcal{X}^{t}}(\int_{\mathcal{X}^{t}%
}\Vert\mathrm{M}(\boldsymbol{\upsilon})\Vert_{q}\mathrm{d}\upsilon_{+}%
^{-})^{2}r(\upsilon_{\circ}^{-}\sqcup\upsilon_{+}^{\circ})\mathrm{d}%
\upsilon_{\circ}^{-}\mathrm{d}\upsilon_{+}^{\circ})^{1/2}$. Then since
\[
\underset{\upsilon\in\mathcal{X}^{t}%
}{\mathrm{ess}\sup}\{s(\upsilon)\Vert\mathrm{M}_{\circ}^{\circ}(\upsilon
)\Vert_{q,t}(r)\}\leq\Vert\mathrm{M}\Vert_{q,t}^{s}(r),
\]
 and since
 \[
 \sup_{\chi\in G_\star(p)}\frac{\|\mathrm{T}\chi\|(p^{-1})}{\|\chi\|(p)}= \sup_{\chi\in G_\star(p)} \frac{|\langle\chi|\mathrm{T}\chi\rangle|}{\|\chi\|(p)\|\chi\|(p)},
  \]
it follows that
$\Vert\mathrm{T}_{t}\Vert_{p}\leq\Vert\mathrm{M}\Vert^s_{q,t}(r)$ for all
$p\geq r^{-1}+q+s^{-1}$.

Now using the definition (\ref{2onee}) and the property
\[
\int_{\mathcal{X}^{t}}\chi(\vartheta)\mathrm{d}\vartheta=\chi(\emptyset
)+\int_{\mathbb{X}^{t}}\mathrm{d}x\int_{\mathcal{X}^{t(x)}}\mathring{\chi}%
(x,\vartheta)\mathrm{d}\vartheta,
\]
where $\mathring{\chi}(x,\vartheta)=\chi(\vartheta\sqcup x)$, it is easy to
see that $[(\mathrm{T}_{t}-\mathrm{T}_{0})\chi](\vartheta
)=[(\boldsymbol{\imath}_{0}^{t}(\mathrm{M})-\mathrm{M}(\emptyset
))\chi](\vartheta)=$
\begin{align*}
&  =\int_{\mathbb{X}^{t}}\mathrm{d}x\sum_{\upsilon_{\circ}^{\circ}\sqcup\upsilon
_{+}^{\circ}\subseteq\vartheta}^{t(\upsilon_{\nu}^{\circ})<t(x)}%
\{\int_{\mathcal{X}^{t(x)}}\mathrm{d}\upsilon_{+}^{-}[\int_{\mathcal{X}%
^{t(s)}}\mathrm{d}\upsilon_{\circ}^{-}(\mathrm{\dot{M}}(\mathbf{x}_{+}%
^{-},\boldsymbol{\upsilon})\mathring{\chi}(\upsilon_{\circ}^{-}\sqcup
\upsilon_{\circ}^{\circ})\\
&  \qquad\qquad\qquad\qquad\qquad\qquad\qquad\qquad\quad+\mathrm{\dot{M}}(\mathbf{x}_{\circ}^{-},\boldsymbol{\upsilon
})\mathring{\chi}(x\sqcup\upsilon_{\circ}^{-}\sqcup\upsilon_{\circ}^{\circ
}))]\}(\vartheta\boldsymbol{\setminus}\upsilon_{\circ}^{\circ
}\boldsymbol{\setminus}\upsilon_{+}^{\circ})\\
& \quad +\sum_{x\in\mathbb{X}^{t}}\int_{\upsilon_{\circ}^{\circ}\sqcup\upsilon
_{+}^{\circ}\subseteq\vartheta}^{t(\upsilon_{\nu}^{\circ})<t(x)}%
\{\int_{\mathcal{X}^{t(x)}}\mathrm{d}\upsilon_{+}^{-}[\int_{\mathcal{X}%
^{t(x)}}\mathrm{d}\upsilon_{\circ}^{-}(\mathrm{\dot{M}}(\mathbf{x}_{+}^{\circ
},\boldsymbol{\upsilon})\mathring{\chi}(\upsilon_{\circ}^{-}\sqcup
\upsilon_{\circ}^{\circ})\\
&  \qquad\qquad\qquad\qquad\qquad\qquad\qquad\qquad\qquad+\mathrm{\dot{M}}(\mathbf{x}_{\circ}^{\circ},\boldsymbol{\upsilon
})\mathring{\chi}(x\sqcup\upsilon_{\circ}^{-}\sqcup\upsilon_{\circ}^{\circ
}))]\}((\vartheta\boldsymbol{\setminus}\upsilon_{\circ}^{\circ
}\boldsymbol{\setminus}\upsilon_{+}^{\circ})\\
&  =\int_{\mathbb{X}^{t}}\mathrm{d}x[\mathrm{D}_{+}^{-}(x)\chi+\mathrm{D}_{\circ}%
^{-}(x)\mathring{\chi}\left(  x\right)  ](\vartheta)+\sum_{x\in \mathbb{X}^{t}%
}[\mathrm{D}_{+}^{\circ}(x)\chi+\mathrm{D}_{\circ}^{\circ}(x)\mathring{\chi
}\left(  x\right)  ](\vartheta\boldsymbol{\setminus}x).
\end{align*}
Consequently, $\mathrm{T}_{t}-\mathrm{T}_{0}=\sum\Lambda_{\mu}^{\nu
}(\mathrm{D}_{\nu}^{\mu},\mathbb{X}^{t})$, where $\Lambda_{\nu}^{\mu}(\mathrm{D}%
,\bigtriangleup)$ are defined in (\ref{2onec}) as operator-valued measures on
$\mathbb{X}$ of operator-functions
\begin{align*}
\lbrack\mathrm{D}_{+}^{\mu}(x)\chi](\vartheta) &  =\sum_{\upsilon_{\circ
}^{\circ}\sqcup\upsilon_{+}^{\circ}\subseteq\vartheta}^{t(\upsilon_{\nu
}^{\circ})<t(x)}\int_{\mathcal{X}^{t(x)}}\mathrm{d}\upsilon_{+}^{-}%
\int_{\mathcal{X}^{t(x)}}\mathrm{d}\upsilon_{\circ}^{-}[\mathrm{\dot{M}%
}(\mathbf{x}_{+}^{\mu},\boldsymbol{\upsilon})\mathring{\chi}(\upsilon_{\circ
}^{-}\sqcup\upsilon_{\circ}^{\circ})](\vartheta_{-}^{\circ}),\\
\lbrack\mathrm{D}_{\circ}^{\mu}(x)\mathring{\chi}\left(  x\right)
](\vartheta) &  =\sum_{\upsilon_{\circ}^{\circ}\sqcup\upsilon_{+}^{\circ
}\subseteq\vartheta}^{t(\upsilon_{\nu}^{\circ})<t(x)}\int_{\mathcal{X}^{t(x)}%
}\mathrm{d}\upsilon_{+}^{-}\int_{\mathcal{X}^{t(x)}}\mathrm{d}\upsilon_{\circ
}^{-}[\mathrm{\dot{M}}(\mathbf{x}_{\circ}^{\mu},\boldsymbol{\upsilon
})\mathring{\chi}(x\sqcup\upsilon_{\circ}^{-}\sqcup\upsilon_{\circ}^{\circ
})](\vartheta_{-}^{\circ}),
\end{align*}
acting on $\chi\in\mathit{G}_{+}$ and $\mathring{\chi}\left(  \upsilon\right)
\in\mathfrak{k}_{\upsilon}^{\otimes}\otimes\mathit{G}_{+}$, where
$\vartheta_{-}^{\circ}=\vartheta\cap\overline{\left(  \upsilon_{\circ}^{\circ
}\sqcup\upsilon_{+}^{\circ}\right)  }=\vartheta\boldsymbol{\setminus
}\upsilon_{\circ}^{\circ}\boldsymbol{\setminus}\upsilon_{+}^{\circ}$.
This can be written in terms of (\ref{2onee}) as
\[
\mathrm{D}_{\nu}^{\mu
}(x)=\boldsymbol{\imath}_{0}^{t}(\mathrm{\dot{M}}(\mathbf{x}_{\nu}^{\mu})).
\]
Because of the inequality $\Vert \mathrm{T}_{t}\Vert
_{p}\leq \Vert \mathrm{M}\Vert _{q,t}^{s}(r)$ for all $p\geq r^{-1}+q+s^{-1}$ we
obtain $\Vert \mathrm{D}_{+}^{-}\Vert _{p,t}^{(1)}\leq \Vert \mathrm{M}\Vert _{q,t}^{s}(r)$,
since $\Vert \mathrm{D}_{+}^{-}(x)\Vert _{p}\leq \Vert \dot{\mathrm{M}}(\mathbf{x}%
_{+}^{-})\Vert _{q,t(x)}^{s}(r)$:
\begin{eqnarray*}
\int_{\mathbb{X}^{t}}\Vert \mathrm{D}_{+}^{-}(x)\Vert _{p}\mathrm{d}x\leq
\int_{\mathbb{X}^{t}}\Vert \dot{\mathrm{M}}(\mathbf{x}_{+}^{-})\Vert _{q,t(x)}^{s}(r)\mathrm{d}%
x
&=&{\int_{\mathbb{X}^{t}}}\mathrm{d}x{\int_{\mathcal{X}^{t(x)}}}\Vert \mathrm{M}_{+}^{-}(x\sqcup
\upsilon )\Vert _{q,t(x)}^{s}(r)\mathrm{d}\upsilon \\
 &=&\int_{\mathcal{X}%
^{t}}\Vert \mathrm{M}_{+}^{-}(\upsilon )\Vert _{q,t}^{s}(r)\mathrm{d}\upsilon
-\Vert \mathrm{M}_{+}^{-}(\emptyset )\Vert _{q,t}^{s}(r) \\
&=&\Vert \mathrm{M}\Vert _{q,t}^{s}(r)-\Vert \mathrm{M}_{+}^{-}(\emptyset )\Vert _{q,t}^{s}(r).
\end{eqnarray*}%
For the estimate of $\mathrm{D}^\circ_+$ we shall require the use of the norm
\[
\|\mathrm{M}^\circ_+\|^s_{q,t}(r):= \left(\int_{\mathcal{X}^t}\Big(\|\mathrm{M}^\circ_+(\upsilon^\circ_+)\|^s_{q,t}(r)\Big)^2 r^\otimes(\upsilon^\circ_+)\mathrm{d}\upsilon^\circ_+\right)^{1/2},
\]
where
\[
\|\mathrm{M}^\circ_+(\upsilon^\circ_+)\|^s_{q,t}(r):= \left(\int_{\mathcal{X}^t}\left(\int_{\mathcal{X}^t}\mathrm{ess}\sup_{\upsilon^\circ_\circ\in \mathcal{X}^t}s(\upsilon^\circ_\circ)\|\mathrm{M}(\boldsymbol{\upsilon})\|_q\mathrm{d}\upsilon^-_+\right)^2 r^\otimes(\upsilon_\circ^-)\mathrm{d}\upsilon_\circ^-\right)^{1/2};
\]
in particular, $\|\mathrm{M}^\circ_+\|^s_{q,t}(r)\leq\|\mathrm{M}\|^s_{q,t}(r)$. So we have
\begin{align*}
\Vert \mathrm{D}_{+}^{\circ }\Vert _{p,t}^{(2)}(r)^2& \leq  \int_{\mathbb{X}^{t}}(\Vert
\dot{\mathrm{M}}^\circ_+(\mathbf{x}_{+}^{\circ })\Vert _{q,t(x)}^{s}(r))^{2}r(x)\mathrm{d}%
x \\
& = \|\mathrm{M}^\circ_+\|^s_{q,t}(r)^2-\|\mathrm{M}^\circ_+(\emptyset)\|^s_{q,t}(r)^2  \leq \Vert \mathrm{M}\Vert _{q,t}^{s}(r)^2.
\end{align*}%
In a similar manner we obtain the estimate for $\mathrm{D}^-_\circ$,
\begin{align*}
\Vert \mathrm{D}_{\circ}^{-}\Vert _{p,t}^{(2)}(r)^2& \leq  \int_{\mathbb{X}^{t}}(\Vert
\dot{\mathrm{M}}_\circ^-(\mathbf{x}^{-}_{\circ })\Vert _{q,t(x)}^{s}(r))^{2}r(x)\mathrm{d}%
x \\
& = \|\mathrm{M}_\circ^-\|^s_{q,t}(r)^2-\|\mathrm{M}_\circ^-(\emptyset)\|^s_{q,t}(r)^2 \leq \Vert \mathrm{M}\Vert _{q,t}^{s}(r)^2.
\end{align*}%
Finally, from $\Vert \mathrm{D}_{\circ }^{\circ }(x)\Vert _{p}\leq \Vert \dot{\mathrm{M}}%
\left( \mathbf{x}_{\circ }^{\circ }\right) \Vert _{q,t(x)}^{s}(r)$ we
similarly obtain
\begin{equation*}
\Vert \mathrm{D}_{\circ }^{\circ }\Vert _{p,t}^{(\infty )}(s)\leq \mathrm{ess}%
\sup_{x\in \mathbb{X}^{t}}\{s(x)\Vert \dot{\mathrm{M}}(\mathbf{x}_{\circ }^{\circ })\Vert
_{q,t(x)}^{s}(r)\}\leq \Vert \mathrm{M}\Vert _{q,t}^{s}(r)
\end{equation*}%
if $p\geq r^{-1}+q+s^{-1}$, which concludes the proof.
\end{proof}
Since $\mathrm{T}_t=\boldsymbol{\imath}^t_0(\mathrm{M})$ the theorem obviously states that $\|\boldsymbol{\imath}^t_0(\mathrm{M})\|_p\leq\|\mathrm{M}\|^s_{q,t}(r)$ for all $p\geq\frac{1}{r}+q+\frac{1}{s}$, and in particular for the case when $\mathrm{M}(\boldsymbol{\upsilon})=0$ if $\sum|\upsilon^\mu_\nu|\neq1$, and $\mathrm{M}(\mathbf{x})=\mathrm{D}(\mathbf{x})$. Then the result of the theorem becomes
\[
\|\boldsymbol{i}^t_0(\mathbf{D})\|_p\leq\|\mathrm{D}\|^s_{q,t}(r)
\]
and now we shall begin to evaluate the quantum stochastic norm (\ref{2onef}) explicitly for the single integrand $\mathrm{D}(\mathbf{x})$.
\[
\|\mathrm{D}\|^s_{q,t}(r)=\|\mathrm{D}^-_+\|^{(1)}_{q,t}+\|\mathrm{D}^-_+(\emptyset)\|^{s}_{q,t}(r)
\]
one may now proceed in a similar manner to find that
\[
\|\mathrm{D}^-_+(\emptyset)\|^{s}_{q,t}(r)=\|\mathrm{D}_\circ^-\|^{(2)}_{q,t}(r)+ \|\mathrm{D}^\circ_+\|^{(2)}_{q,t}(r) +\|\mathrm{D}_\circ^\circ\|^{(\infty)}_{q,t}(s),
\]
and thus we have recovered that inequality stated at the beginning of this section.

\section{Adapted and Q-Adapted QS Integrals}

The quantum-stochastic integral \textup{(\ref{2onee})} constructed in
{\cite{Be91}}, as well as its single variations \textup{(\ref{2oneb})}
introduced in {\cite{Be90c}}, are defined explicitly and do not require
that the functions $\mathrm{M}$ and $\mathbf{D}$ under the integral be
adapted. By virtue of the continuity we have proved above, they can be
approximated in the inductive convergence by the sequence of integral sums
$\boldsymbol{\imath}_{0}^{t}(\mathrm{M}_{n})$, $\boldsymbol{i}_{0}%
^{t}(\mathbf{D}_{n})$ corresponding to step measurable operator functions
$\mathrm{M}_{n}$ and $\mathbf{D}_{n}$ if the latter converge inductively to
$\mathrm{M}$ and $\mathbf{D}$ in the poly-norm \textup{(\ref{2onef})}.
In fact, if there exist functions $r$, $s$ with $r^{-1}$, $s^{-1}%
\in\mathfrak{p}_{0}$ and $q\in\mathfrak{p}_{1}$ such that $\Vert\mathrm{M}%
_{n}-\mathrm{M}\Vert_{q,t}^{s}(r)\rightarrow0$, then there also exists a
function $p\in\mathfrak{p}_{1}$ such that $\Vert\boldsymbol{\imath}_{0}%
^{t}(\mathrm{M}_{n}-\mathrm{M})\Vert_{p}\rightarrow0$, and we have $p\geq
r^{-1}+q+s^{-1}$ by the inequality \textup{(\ref{2oneh})}, which implies the
inductive convergence $\boldsymbol{\imath}_{0}^{t}(\mathrm{M}_{n}%
)\rightarrow\boldsymbol{\imath}_{0}^{t}(\mathrm{M})$ as a result of the
linearity of $\boldsymbol{\imath}_{0}^{t}$.

Let $\mathrm{Q}:x\mapsto \mathfrak{L}\left( \mathfrak{k}_{x}\right) $ be a
measurable adjointable operator-valued function with $q$-contractive values $\left\Vert
\mathrm{Q}\left( x\right) \right\Vert \leq q\left( x\right) $ with respect
to a positive function $q\in \mathfrak{p}$. We shall say that the integrand $%
\mathbf{D}(x)$ is $\mathrm{Q}$-adapted if it has the product form
\begin{equation}
\mathrm{D}%
_{\nu }^{\mu }(x)=\mathrm{D}_{\nu }^{\mu }(x)^{t\left( x\right) }\otimes
\mathrm{Q}_{t(x)}^{\otimes }
\end{equation}
with respect to the Fock split $\mathcal{G}_\ast=$
$\mathcal{G}_\ast^{t\left( x\right) }\otimes \mathcal{F}_{\ast[t\left( x\right) }$
corresponding to $\mathbb{X}=\mathbb{X}^{t\left( x\right) }\sqcup \mathbb{X}_{[t\left( x\right) }$,
where $\mathrm{D}_{\nu }^{\mu }(x)^{t\left( x\right) }\equiv\mathrm{K}_\nu^\mu(x)$ is the restricted
action $\mathit{G}_{+}^{t\left( x\right) }\rightarrow \mathit{G}%
_{-}^{t\left( x\right) }$ of $\mathrm{D}_{\nu }^{\mu }(x)$ corresponding to
the projections $\mathrm{P}_{t(x)}\mathrm{D}^\mu_\nu(x)\mathrm{P}_{t(x)}$ where  the projection $\mathrm{P}_t$ is given on $\mathcal{G}_\ast$ by
 \[
 \big[\mathrm{P}_t\chi\big](\vartheta)=\chi(\vartheta^t)\delta_\emptyset(\varkappa),\quad \mathrm{P}_0\equiv\mathrm{P}_\emptyset\]
 where $\varkappa=\vartheta\setminus\vartheta^t$, and $\delta_\emptyset(\vartheta)=1$ if $\vartheta$ is empty, and is otherwise zero.
 We can now write the $\mathrm{Q}$-adapted QS\
integrals as a more explicit form of the general QS\ integrals (\ref{2onec})
\begin{align}
\qquad\;\lbrack\Lambda_{-}^{+}(\mathrm{D}_{+}^{-},\boldsymbol{\bigtriangleup}%
)\chi](\vartheta) &  =\int_{\boldsymbol{\bigtriangleup}}[\mathrm{K}_{+}%
^{-}(x)\otimes\mathrm{Q}^\otimes(\varkappa)\mathring{\chi}(\varkappa)](\vartheta^{t(x)})\mathrm{d}x,\quad \;\;\;\;\;\;x\notin\vartheta\nonumber\\
\;\lbrack\Lambda_{\circ}^{+}(\mathrm{D}_{+}^{\circ},\boldsymbol{\bigtriangleup
})\chi](\vartheta) &  =\sum_{x\in\boldsymbol{\bigtriangleup}\cap\vartheta
}[\mathrm{K}_{+}^{\circ}(x)\otimes\mathrm{Q}^\otimes(\varkappa)\mathring{\chi}(\varkappa)] (\vartheta^{t(x)}),\nonumber\\
\qquad\;\lbrack\Lambda_{-}^{\circ}(\mathrm{D}_{\circ}^{-},\boldsymbol{\bigtriangleup
})\chi](\vartheta) &  =\int_{\boldsymbol{\bigtriangleup}}[\mathrm{K}_{\circ
}^{-}(x)\otimes\mathrm{Q}^\otimes(\varkappa)\mathring{\chi}(x\sqcup\varkappa)](\vartheta^{t(x)})\mathrm{d}x,\quad x\notin\vartheta
\nonumber\\
\lbrack\Lambda_{\circ}^{\circ}(\mathrm{D}_{\circ}^{\circ}%
,\boldsymbol{\bigtriangleup})\chi](\vartheta) &  =\sum_{x\in
\boldsymbol{\bigtriangleup}\cap\vartheta}[\mathrm{K}_{\circ}^{\circ
}(x)\otimes\mathrm{Q}^\otimes(\varkappa)\mathring{\chi}(x\sqcup\varkappa)](\vartheta^{t(x)}).%
\end{align}
Here $\mathrm{Q}_{}^{\otimes }\left(
\varkappa \right) =\otimes _{x\in \varkappa }\mathrm{Q}\left( x\right) $ for
any $\varkappa =\vartheta_{t(x)}\in \mathcal{X}_{t\left( x\right) }$ and $\mathring{\chi}%
(\varkappa ,\vartheta ^{t})=\chi (\vartheta ^{t}\sqcup \varkappa )$.

In the next proposition we look at the some of the properties of a special class of $\mathrm{Q}$-adapted integrals. In particular, we consider the $g$-commutator of the creation and annihilation integrals $\Lambda(\mathbf{D}^\circ_+)$ and $\Lambda(\mathbf{D}^-_\circ)$, given by
\begin{equation}
\big[\Lambda(\mathbf{D}^-_\circ),\Lambda(\mathbf{D}^\circ_+)\big]_{g}:= \Lambda(\mathbf{D}^-_\circ)\Lambda(\mathbf{D}^\circ_+)-g\big(\Lambda(\mathbf{D}^\circ_+) \Lambda(\mathbf{D}_\circ^-)\big), \label{qcom}
\end{equation}
where $g$ is considered to be an element of $ L^\infty(\mathbb{X}\times\mathbb{X})$ in the sense that its action on $\Lambda(\mathbf{D}^\circ_+) \Lambda(\mathbf{D}_\circ^-)$ is defined as
\begin{equation}
g\big(\Lambda(\mathbf{D}^\circ_+) \Lambda(\mathbf{D}_\circ^-)\big)=\iint\Lambda^+_\circ(\mathrm{d}z)\mathrm{D}^\circ_+(z)g(z,x)\mathrm{D}^-_\circ(x) \Lambda^\circ_-(\mathrm{d}x),
\end{equation}
with respect to the identifications
\[
\Lambda(\mathbf{D}^\circ_+(z),\mathrm{d}z)=\Lambda^+_\circ(\mathrm{d}z)\mathrm{D}^\circ_+(z) ,\quad\Lambda(\mathbf{D}_\circ^-(x),\mathrm{d}x)=\mathrm{D}^-_\circ(x) \Lambda^\circ_-(\mathrm{d}x),
\]
$x,z\in\mathbb{X}$.

\begin{proposition}
Let $\Lambda(\mathbf{D}_\circ^-)$ and $\Lambda(\mathbf{D}^\circ_+)$ be $\mathrm{Q}$-adapted quantum stochastic integrals  of the form
\begin{align*}
[\Lambda(\mathbf{D}^-_\circ,\mathbb{X}^t)\chi](\vartheta)&= \int_{\mathbb{X}^t} \mathrm{K}_\circ^-(x,\vartheta^{t(x)})\otimes{\mathrm{Q}^\ast}^{\otimes}(\vartheta_{t(x)})\chi(\vartheta\sqcup x)\mathrm{d}x,\\
[\Lambda(\mathbf{D}^\circ_+,\mathbb{X}^t)\chi](\vartheta)&=\sum_{z\in\vartheta^t}
\mathrm{K}^\circ_+(z,\vartheta^{t(z)})\otimes\mathrm{Q}^{\otimes}(\vartheta_{t(z)}) \chi(\vartheta\setminus z),
\end{align*}
 such that $\mathrm{K}^-_\circ(x)$ and $\mathrm{K}^\circ_+(x)$ are diagonal on $\mathcal{X}^{t(x)}$ in the sense that  $\big[\mathrm{K}^\circ_+(x)\chi\big](\vartheta^{t(x)})= \mathrm{K}^\circ_+(x,\vartheta^{t(x)})\chi(\vartheta^{t(x)})$ and $\big[\mathrm{K}_\circ^-(x)\dot{\chi}(x)\big](\vartheta^{t(x)}) =\mathrm{K}_\circ^-(x,\vartheta^{t(x)})\chi(x\sqcup\vartheta^{t(x)})$,
 where $\mathrm{Q}$  is $q$-contractive normal operator $[\mathrm{Q},\mathrm{Q}^\ast]=0$. Further, suppose that the operator $\mathrm{K}^-_\circ(x)$ commutes with $\mathrm{K}^\circ_+(z)$ in $\mathcal{G}_\ast$ such that $[\mathrm{K}^-_\circ(x),\mathrm{K}^\circ_+(z)]=0$ for all $x\neq z\in\mathbb{X}^t$, and that $[\mathrm{K}^-_\circ(x),\mathrm{Q}^\otimes]=0=[\mathrm{K}^\circ_+(x),\mathrm{Q}^\otimes]$ in $\mathcal{G}_\ast$, then the $g$-commutator \textup{(\ref{qcom})}
satisfies the equation
\begin{equation}
\big[\Lambda(\mathbf{D}_\circ^-,\mathbb{X}^t),\Lambda(\mathbf{D}^\circ_+,\mathbb{X}^t)\big]_{g} =\int_{\mathbb{X}^t}\mathrm{K}^-_\circ(x)\mathrm{K}^\circ_+(x)\mathrm{d}x\otimes(\mathrm{Q}^\ast\mathrm{Q})^\otimes
\end{equation}
in $\mathcal{G}_\ast^t\otimes\mathcal{F}_{\ast t}$
if we have $\mathrm{K}^-_\circ(x,\vartheta^{t(x)}) \mathrm{Q}( x)=g(z,x)\mathrm{K}^-_\circ(x,\vartheta^{t(x)}\setminus z)\otimes\mathrm{I}_z$ for all $z\in\vartheta^{t(x)}$, and
${\mathrm{Q}^\ast}(z)\mathrm{K}^\circ_+(z,\vartheta^{t(z)}\sqcup x) =\mathrm{I}_x\otimes \mathrm{K}_+^\circ(z,\vartheta^{t(z)})g(z,x)$
 for all
 $x<z$ with $x\notin\vartheta^{t(z)}$.
\end{proposition}
\begin{proof}
 First we calculate \newline$[\Lambda(\mathbf{D}_\circ^-,\mathbb{X}^t)\Lambda(\mathbf{D}^\circ_+,\mathbb{X}^t)\chi](\vartheta)- \big[\int_{\mathbb{X}^t}\mathrm{K}^-_\circ(x)\mathrm{K}^\circ_+(x) \mathrm{d}x\otimes(\mathrm{Q}^\ast\mathrm{Q})^\otimes\chi\big](\vartheta)=$
\begin{align*}
=&\int_{\mathbb{X}^t}\sum_{z\in\vartheta^{t(x)}}\mathrm{K}^-_\circ(x,\vartheta^{t(x)}) \mathrm{K}^\circ_+(z,\vartheta^{t(z)}) \mathrm{Q}^{\otimes}(\vartheta^{t(x)}_{t(z)}\sqcup x) (\mathrm{Q}^\ast\mathrm{Q})^{\otimes}(\vartheta_{t(x)}) \chi(x\sqcup\vartheta\setminus z)\mathrm{d}x\\
&+\sum_{z\in\vartheta}\int_{\mathbb{X}^{t(z)}}\mathrm{K}^-_\circ(x,\vartheta^{t(x)}) {\mathrm{Q}^\ast}^{\otimes}(\vartheta^{t(z)}_{t(x)}\sqcup z) \mathrm{K}^\circ_+(z,\vartheta^{t(z)}\sqcup x) (\mathrm{Q}^\ast\mathrm{Q})^{\otimes}(\vartheta_{t(z)})\chi(x\sqcup\vartheta\setminus z)\mathrm{d}x,
\end{align*}
where $\vartheta_s^t=\vartheta\cap\mathbb{X}_s^t$, and then we calculate $g[\Lambda(\mathbf{D}_+^\circ,\mathbb{X}^t)\Lambda(\mathbf{D}_\circ^- ,\mathbb{X}^t)\chi](\vartheta)=$
\begin{align*}
=&\int_{\mathbb{X}^t}\sum_{z\in\vartheta^{t(x)}}g(z,x))\mathrm{K}^\circ_+(z,\vartheta^{t(z)}) \mathrm{Q}^{\otimes}(\vartheta^{t(x)}_{t(z)})\mathrm{K}^-_\circ(x,\vartheta^{t(x)}\setminus z)(\mathrm{Q}^\ast\mathrm{Q})^{\otimes}(\vartheta_{t(x)}) \chi(x\sqcup\vartheta\setminus z)\mathrm{d}x\\
&+\sum_{z\in\vartheta}\int_{\mathbb{X}^{t(z)}}g(z,x)\mathrm{K}_+^\circ(z,\vartheta^{t(z)}) \mathrm{K}^-_\circ(x,\vartheta^{t(x)}) {\mathrm{Q}^\ast}^{\otimes}(\vartheta^{t(z)}_{t(x)}) (\mathrm{Q}^\ast\mathrm{Q})^{\otimes}(\vartheta_{t(z)})\chi(x\sqcup\vartheta\setminus z)\mathrm{d}x,
\end{align*}
thus when equating the integrands for these two formula, using the commutativity of the factors, we find that for $z\in\vartheta^{t(x)}$
\[
\mathrm{K}^-_\circ(x,\vartheta^{t(x)}) \mathrm{Q}( x)=g(z,x)\mathrm{K}^-_\circ(x,\vartheta^{t(x)}\setminus z)
\]
and for $x<z$ we find
\[
 {\mathrm{Q}^\ast}(z)\mathrm{K}^\circ_+(z,\vartheta^{t(z)}\sqcup x) = \mathrm{K}_+^\circ(z,\vartheta^{t(z)})g(z,x).
\]
\end{proof}
Notice that when $g(z,x)$ is separable we have $g(z,x)\mathrm{I}_x=h(z)\mathrm{Q}(x)$ for $z<x$, and $g(z,x)\mathrm{I}_z=\mathrm{Q}^\ast(z)f(x)$ when $z>x$, and also note that if $h=f^\ast$ then $g$ is self-adjoint, $g^\ast(x,z)=g(z,x)$, having the form
\begin{equation}g(z,x)=h^\ast(z)q(x)1_{t(x)}(z)+q^\ast(z)h(x)1_{ t(z)}(x),\quad z\neq x\end{equation}
where $1_{t(x)}(z)=1$ if $z<x$ otherwise 0.
Further, from this separability of $g$ it follows that $\mathrm{K}^-_\circ(x,\upsilon)=\mathrm{k}^-_\circ(x) \otimes h^\otimes(\upsilon)\mathrm{I}_\upsilon$ for all $\upsilon\in\mathcal{X}^{t(x)}$, and that $\mathrm{K}^\circ_+(z,\upsilon)=\mathrm{k}^\circ_+(z)\otimes f^\otimes(\upsilon)\mathrm{I}_\upsilon$ for all $\upsilon\in\mathcal{X}^{t(z)}$. 
Consider for example the case when $g(z,x)=\exp\{\mathrm{i}p(x-z)\}$, where $p:\mathbb{X}\rightarrow\mathbb{R}$ is linear functional $p(x-z)=p(x)-p(z)$.
\begin{corollary}
Suppose that  $g(x,z)$, $x\neq z$, is a constant $\mathbb{R}$-valued function of the form
 \[
 g(z,x)=1(z)q(x)1_{t(x)}(z)+q(z)1(x)1_{ t(z)}(x),\quad q(x)=q\in\mathbb{R}\;\forall\;x
  \]
  where $1_t(x)=1$ if $t(x)<t$ and otherwise zero, then we define the $q$-commutator
\begin{equation}
\big[\Lambda(\mathbf{D}^-_\circ),\Lambda(\mathbf{D}^\circ_+)\big]_{q}= \Lambda(\mathbf{D}^-_\circ)\Lambda(\mathbf{D}^\circ_+)-q\Lambda(\mathbf{D}^\circ_+) \Lambda(\mathbf{D}_\circ^-),
\end{equation}
and we recover the commutation relations of the Bosonic and Fermionic quantum fields under the above conditions with $\mathrm{Q}(x)=q\mathrm{I}$  when we respectively set $q=+1$ and $q=-1$, such that
\[
[\Lambda(\mathbf{D}_\circ^-,\mathbb{X}^t),\Lambda(\mathbf{D}^\circ_+,\mathbb{X}^t)]_{\pm1} =\int_{\mathbb{X}^t}\mathrm{K}^-_\circ(x)\mathrm{K}^\circ_+(x)\mathrm{d}x\otimes\mathrm{I}^\otimes
\]
and thus we have the representations of the Bosonic and Fermionic quantum fields in the Guichardet-Fock space as $\pm\mathrm{I}$-adapted QS integrals.

Further,
the trivial commutator of the {monotonic quantum field} is obtained when $\mathrm{Q}=\mathrm{O}$, corresponding also to $q=0$, such that
\[
[\Lambda(\mathbf{D}_\circ^-,\mathbb{X}^t),\Lambda(\mathbf{D}^\circ_+,\mathbb{X}^t)]_{0} =\int_{\mathbb{X}^t}\mathrm{K}^-_\circ(x)\mathrm{K}^\circ_+(x)\mathrm{d}x\otimes\mathrm{O}^\otimes
\]
where we have made use of the identification $\mathrm{D}^\mu_\nu(\emptyset)=0$. Such $\mathrm{O}$-adapted monotonic calculus, also called the {vacuum adapted} calculus as $\mathrm{O}^\otimes=\mathrm{P}_{\emptyset}$, was also studied by Belton in \cite{Bel98}.
\end{corollary}

Obviously the QS integral $\mathrm{X}\left( t\right) =\boldsymbol{i}%
_{0}^{t}\left( \mathbf{D}\right) $ of any $\mathrm{Q}$-adapted integrand $%
\mathbf{D}\left( x\right) $ is an operator-valued \textrm{Q}-adapted process in
the sense that $\mathrm{X}\left( t\right) =\mathrm{X}^{t}\otimes \mathrm{Q}_t%
^{\otimes }$. The approximation of this integral in the class of adapted step functions,  leads by continuity to the usual definition of the quantum-stochastic integral
$\boldsymbol{i}_{0}^{t}(\mathbf{D})$ which was given by Hudson and
Parthasarathy corresponding to the $\mathrm{I}$-adapted case with $\mathbb{X}=\mathbb{R}_{+}$ such that $x\equiv t(x)$. That is
the weak limit of integral sums
\[
\boldsymbol{i}_{0}^{t}(\mathbf{D}_{n})=\int_{0}^{t}\Lambda(\mathbf{D}%
_{n},\mathrm{d}x)=\sum_{j=1}^{n}\mathrm{D}_{\nu}^{\mu}(x_{j})\mathrm{A}_{\mu
}^{\nu}(\bigtriangleup_{j}),
\]
where $\mathbf{D}(x_{j})=\mathbf{D}_{n}(x)$ for $x\in\lbrack x_{j},x_{j+1})$ is
an adapted approximation corresponding to the decomposition $\mathbb{R}%
_{+}=\sum_{j=1}^{n}\bigtriangleup_{i}$ into the intervals $\bigtriangleup
_{j}=[x_{j},x_{j+1})$ given by the chain $x_{0}=0<x_{1}<\cdots<x_{n}%
<x_{n+1}=\infty$, and $\mathrm{D}_{\nu}^{\mu}(x)\mathrm{A}_{\mu}^{\nu
}(\bigtriangleup)$ is the sum of the operators (\ref{2onec}) with functions
$\mathrm{D}_{\nu}^{\mu}(x)$ constant on $\bigtriangleup$ which can therefore
be pulled out in front of the integrals $\Lambda_{\mu}^{\nu} $ such that $\mathrm{A}^\nu_\mu(\triangle)=\int_{\triangle}\Lambda^\nu_\mu(\mathrm{d}x)$.

In particular, for $\mathrm{D}_{+}^{-}=0=\mathrm{D}_{\circ}^{\circ}$ and
$\mathrm{D}_{\circ}^{-}=k\otimes\mathrm{\hat{1}}=\mathrm{D}_{+}^{\circ}$,
where $\mathrm{\hat{1}}=\mathrm{I}^{\otimes}$ is the unit operator in
$\mathcal{F}_{\ast}$ and $k(x)$ is a scalar locally square integrable function
corresponding to the case $\mathfrak{k}_{x}=\mathbb{C}=\mathfrak{h}$, we
obtain the It\^{o} definition of the Wiener integral
\[
\dot{I}_{0}^{t}(k)=\int_{0}^{t}k(x)w(\mathrm{d}x),\quad\int_{0}^{t}%
k(x)\widehat{w}(\mathrm{d}x)=\boldsymbol{i}_{0}^{t}(\mathbf{D})
\]
with respect to the stochastic measure $w(\bigtriangleup)$, $\bigtriangleup
\in\mathfrak{F}_{\mathbb{X}}$ on $\mathbb{R}_{+}$, represented in $\mathcal{G}_{\ast
}=\mathcal{F}_{\ast}$ by the operators $\widehat{w}(\bigtriangleup
)=\mathrm{A}_{\circ}^{+}(\bigtriangleup)+\mathrm{A}_{-}^{\circ}(\bigtriangleup)$. We
also note that the multiple integral \textup{(\ref{2onee})} in the trivially
adapted case $\mathrm{M}(\boldsymbol{\upsilon})=M(\boldsymbol{\upsilon
})\otimes\mathrm{I}^{\otimes}$ defines the Fock representation of the
generalized Maassen-Meyer kernels {\cite{EvaH88}, {\cite{Mey87}} and in the
case
\[
M(\boldsymbol{\upsilon})=m({\normalsize \upsilon}_{\circ}^{-}\sqcup
{\normalsize \upsilon}_{+}^{\circ})\delta_{\emptyset}({\normalsize \upsilon
}_{+}^{-})\delta_{\emptyset}({\normalsize \upsilon}_{\circ}^{\circ}%
),\quad\delta_{\emptyset}({\normalsize \upsilon})=%
\begin{cases}
1, & {\normalsize \upsilon}=\emptyset,\\
0, & {\normalsize \upsilon}\neq\emptyset
\end{cases}
\]
it leads to the multiple stochastic integrals $\boldsymbol{\imath}_{0}%
^{t}(\mathrm{M})=\widehat{I}_{0}^{t}(m)$,
\[
I_{0}^{t}(m)=\sum_{n=0}^{\infty}\quad\idotsint\limits_{0\leq t_{1}%
<\cdots<t_{n}<t}m(x_{1},\ldots,x_{n})w(\mathrm{d}x_{1})\ldots w(\mathrm{d}%
x_{n})
\]
of the generalized functions $m\in\bigcup_{r^{-1}\in\mathfrak{p}_{0}%
}\mathit{G}_{\star}(r)$, that is, to the Hida distributions {\cite{Hid80}, {\cite{PotS89}} of the Wiener measure $w(\bigtriangleup)$ represented as
$\widehat{w}(\bigtriangleup)$. Thus, we can consider the trivially adapted QS
multiple integrals $\boldsymbol{\imath}_{0}^{t}(M\otimes\mathrm{I}^\otimes)$ as quantum Hida
operator-distributions whose properties are described in the following corollary when $\mathrm{Q}=\mathrm{I}$.

\begin{corollary}
\label{C 1} Suppose that $\mathrm{M}(\boldsymbol{\upsilon}%
)=M(\boldsymbol{\upsilon})\otimes\mathrm{Q}^{\otimes}$ where $\|\mathrm{Q}^\otimes\|_q=1$, i.e. the
operator-function $\mathrm{M}$ is defined by the $q$-contractive ampliation of the $\star$-kernel $M$
with $\Vert M\Vert_{t}^{s}(r)<\infty$, where
\[
M%
\begin{pmatrix}
\upsilon_{+}^{-}, & \upsilon_{\circ}^{-}\\
\upsilon_{+}^{\circ}, & \upsilon_{\circ}^{\circ}%
\end{pmatrix}
:\mathfrak{k}^{\otimes}\left(  \upsilon_{\circ}^{-}\sqcup\upsilon_{\circ
}^{\circ}\right)\otimes\mathfrak{h}  \rightarrow\mathfrak{k}^{\otimes}\left(
\upsilon_{\circ}^{\circ}\sqcup\upsilon_{+}^{\circ}\right)\otimes\mathfrak{h}  ,
\]
and $\Vert M\Vert_{t}^{s}(r)<\infty$ for all $t\in\mathbb{R}_+$, and for some $r,s$ with $r^{-1},s^{-1}\in\mathfrak{p}_0$, where
\[
\Vert M\Vert_{t}^{s}(r)=\int_{\mathcal{X}^{t}}\mathrm{d}\upsilon_{+}%
^{-}\Big(\int_{\mathcal{X}^{t}}\mathrm{d}\upsilon_{+}^{\circ}\int
_{\mathcal{X}^{t}}\mathrm{d}\upsilon_{\circ}^{-}\mathrm{ess}\sup
_{\upsilon_{\circ}^{\circ}\in\mathcal{X}^{t}}\{s(\upsilon_{\circ}^{\circ
})\Vert M(\boldsymbol{\upsilon})\Vert\}^{2}r(\upsilon_{+}^{\circ}%
\sqcup\upsilon_{\circ}^{-})\Big)^{1/2}%
\]
and $r(\upsilon)=\prod_{x\in\upsilon
}r(x)$, $s(\upsilon)=\prod_{x\in\upsilon}s(x)$. Then the multiple integral \textup{(\ref{2onee})} defines a
$\mathrm{Q}$-adapted family $\mathrm{T}_{t}$, $t\in\mathbb{R}_{+}$, of $p$-bounded
operators \[\mathrm{T}_{t}=\boldsymbol{\imath}_{0}^{t}({M}\otimes\mathrm{Q}^\otimes),\quad
\Vert\mathrm{T}_{t}\Vert_{p}\leq\Vert M\Vert_{t}^{s}(r)\] for $p\geq
r^{-1}+q+s^{-1}$, with bounded $\mathrm{Q}$-adapted quantum-stochastic derivatives
\[\mathrm{D}_{\nu}^{\mu}(x)=\boldsymbol{\imath}_{0}^{t(x)}(\dot{M}%
(\mathbf{x}_{\nu}^{\mu})\otimes\mathrm{Q}^\otimes)\equiv\mathrm{K}^\mu_\nu(x)\otimes\mathrm{Q}^\otimes_{t(x)}.\]
\end{corollary}
\begin{proof}
Since $\|M(\boldsymbol{\upsilon})\otimes\mathrm{Q}^\otimes\|_q\leq\|M(\boldsymbol{\upsilon})\|$ the result immediately follows from the inequality $\|\boldsymbol{\imath}^t_0(M\otimes\mathrm{Q}^\otimes) \|_p\leq\|M\otimes\mathrm{Q}^\otimes\|^s_{q,t}(r)\leq\|M\|^s_t(r)$ for $p\geq
r^{-1}+q+s^{-1}$.
\end{proof}

\end{document}